\documentclass{article}
\usepackage[nonatbib,final]{neurips_2020}

\usepackage[utf8]{inputenc} % allow utf-8 input
\usepackage[T1]{fontenc}    % use 8-bit T1 fonts
\usepackage{hyperref}       % hyperlinks
\usepackage{url}            % simple URL typesetting
\usepackage{booktabs}       % professional-quality tables
\usepackage{amsfonts}       % blackboard math symbols
\usepackage{nicefrac}       % compact symbols for 1/2, etc.
\usepackage{microtype}      % microtypography
\usepackage{lipsum}
\usepackage[numbers]{natbib}
\usepackage{graphicx}

\usepackage[small]{caption}
\usepackage{algorithm}
\usepackage[noend]{algpseudocode}
\usepackage{amssymb}
\usepackage{mathtools}
\usepackage{amsmath}
\usepackage{tikz}
\usetikzlibrary{calc}
\usepackage{todonotes}
\usepackage{dsfont}
\usepackage{enumitem}
\usepackage{tikz}
\usepackage{amsthm}
\usepackage{thm-restate}
\usepackage[capitalise,noabbrev]{cleveref}
\usepackage{booktabs}
\usepackage{multirow}
\usepackage{newfloat}
\usepackage{wrapfig}
\usepackage{dsfont}
\usepackage{comment}
\usepackage{multirow}
\usepackage{mleftright}
\usepackage{mdframed}
\usepackage{bm}
\usepackage{multirow}

\newtheorem{lemma}{Lemma}
\newtheorem{remark}{Remark}

\newtheorem{definition}{Definition}

\newcommand{\pl}{\mathcal{P}}
\newcommand{\I}{\mathcal{I}}
\newcommand{\X}{\mathcal{X}}

\newcommand{\one}{\mathds{1}}
\newcommand{\icfr}{{ICFR}}
\newcommand{\rme}{\mathcal{R}^\textsc{ext}}
\newcommand{\rmi}{\mathcal{R}^\textsc{int}}

\newcommand\numberthis[1]{\addtocounter{equation}{1}\tag{\theequation}\label{#1}}
\newcommand*\circled[1]{\tikz[baseline=(char.base)]{
            \node[shape=circle,draw,inner sep=1pt] (char) {\scriptsize #1};}}
\newcommand{\defeq}{\mathrel{:\mkern-0.25mu=}}

\makeatletter
\newcommand{\vast}{\bBigg@{4}}
\makeatother

\title{No-Regret Learning Dynamics for Extensive-Form Correlated Equilibrium}

\author{
	Andrea~Celli\thanks{Equal contribution.}\\%\thanks{The work was conducted while the author was a postdoc at Politecnico di Milano.}\\
	Politecnico di Milano\\
	\texttt{andrea.celli@polimi.it}
	\And
	Alberto~Marchesi$^\ast$\\
	Politecnico di Milano\\
	\texttt{alberto.marchesi@polimi.it}
	\And
	Gabriele Farina$^*$\\
	Carnegie Mellon University\\
	\texttt{gfarina@cs.cmu.edu}
	\And
	Nicola Gatti\\
	Politecnico di Milano\\
	\texttt{nicola.gatti@polimi.it}
}

\begin{document}
    \maketitle

    \begin{abstract}
	The existence of simple, uncoupled no-regret dynamics that converge to correlated equilibria in normal-form games is a celebrated result in the theory of multi-agent systems. Specifically, it has been known for more than 20 years that when all players seek to minimize their \emph{internal} regret in a repeated normal-form game, the empirical frequency of play converges to a normal-form correlated equilibrium. Extensive-form (that is, tree-form) games generalize normal-form games by modeling both sequential and simultaneous moves, as well as private information.
	Because of the sequential nature and presence of partial information in the game, extensive-form correlation has significantly different properties than the normal-form counterpart, many of which are still open research directions.
	%%
	% Extensive-form correlated equilibrium (EFCE) is a relatively recent correlated solution concept for extensive-form games. In this paper, we give the first uncoupled no-regret dynamics that converge to the set of EFCEs in $n$-player general-sum extensive-form games with perfect recall. The existence of such dynamics strengthens EFCE's position as the natural extensive-form counterpart to normal-form correlated equilibrium.
	%	
	Extensive-form correlated equilibrium (EFCE) has been proposed as the natural extensive-form counterpart to normal-form correlated equilibrium. However, it was currently unknown whether EFCE emerges as the result of uncoupled agent dynamics.
	%
%	In this context, the {\em extensive-form correlated equilibrium} (EFCE) is a relatively recent solution concept which is arguably considered the most appropriate notion of correlation for extensive-form games.
%	%
%	However, the lack of simple learning dynamics converging to an EFCE has been limiting the applicability of this solution concept in practice.
	%
	In this paper, we give the first uncoupled no-regret dynamics that converge to the set of EFCEs in $n$-player general-sum extensive-form games with perfect recall.
	First, we introduce a notion of \emph{trigger regret} in extensive-form games, which extends that of internal regret in normal-form games. When each player has low trigger regret, the empirical frequency of play is close to an EFCE. Then, we give an efficient no-trigger-regret algorithm. Our algorithm decomposes trigger regret into local subproblems at each decision point for the player, and constructs a global strategy of the player from the local solutions at each decision point.
    \end{abstract}

    \section{Introduction}

The {\em Nash equilibrium} (NE)~\citep{nash1950equilibrium} is the most common notion of rationality in game theory, and its computation in two-player, zero-sum games has been the flagship computational challenge in the area at the interplay between computer science and game theory (see, {\em e.g.}, the landmark results in heads-up no-limit poker by~\citet{brown2017superhuman} and~\citet{moravvcik2017stack}).
The assumption underpinning NE is that the interaction among players is fully {\em decentralized}.
Therefore, an NE is a distribution on the {\em uncorrelated} strategy space ({\em i.e.}, a product of independent distributions, one per player).
A competing notion of rationality is the {\em correlated equilibrium} (CE) proposed by~\citet{aumann1974subjectivity}.
A {\em correlated strategy} is a general distribution over joint action profiles and it is customarily modeled via a trusted external {\em mediator} that draws an action profile from this distribution, and privately recommends to each player her component.
A correlated strategy is a CE if no player has an incentive to choose an action different from the mediator's recommendation, because, assuming that all other players also obey, the suggested strategy is the best in expectation.

Many real-world strategic interactions involve more than two players with arbitrary ({\em i.e.}, general-sum) utilities.
In these settings, the notion of NE presents some weaknesses which render the CE a  natural solution concept: (i) computing an NE is an intractable problem, being \textsf{PPAD}-complete even in two-player games~\citep{chen2006settling,daskalakis2009complexity}; (ii) the NE is prone to equilibrium selection issues; and (iii) the social welfare that can be attained via an NE may be significantly lower than what can be achieved via a CE~\citep{koutsoupias1999worst,roughgarden2002bad}.
Moreover, in normal-form games, the notion of CE arises from simple learning dynamics in senses that NE does not~\citep{hart2000simple,cesa2006prediction}.

The notion of {\em extensive-form correlated equilibrium} (EFCE) by~\citet{von2008extensive} is a natural extension of the CE to the case of sequential strategic interactions.
In an EFCE, the mediator draws, before the beginning of the sequential interaction, a recommended action for each of the possible decision points ({\em i.e.}, {\em information sets}) that players may encounter in the game, but she does not immediately reveal recommendations to each player.
Instead, the mediator incrementally reveals relevant individual moves as players reach new information sets.
At any decision point, the acting player is free to defect from the recommended action, but doing so comes at the cost of future recommendations, which are no longer issued if the player deviates.

\paragraph{Original contributions} 
We focus on general-sum extensive-form games with an arbitrary number of players (including the {\em chance player}). 
In this setting, the problem of computing a feasible EFCE can be solved in polynomial time in the size of the game tree~\citep{huang2008computing} via a variation of the {\em Ellipsoid Against Hope} algorithm~\citep{papadimitriou2008,jiang2015polynomial}.
However, in practice, this approach cannot scale beyond toy problems.
Therefore, the following question remains open: {\em is it possible to devise simple dynamics leading to a feasible EFCE?}
In this paper, we show that the answer is positive. 
To do so, we define an EFCE via the notion of {\em trigger agent}~\citep{gordon2008no,dudik2009sampling}.
Then, we define the notion of \emph{trigger regret}, {\em i.e.}, a notion of {\em internal regret} suitable for extensive-form games.
We provide an algorithm, which we call \icfr, that minimizes trigger agent regrets via the decomposition of these regrets locally at each information set.
In order to do so, \icfr~instantiates an internal regret minimizer and multiple external regret minimizers for each information set.
We show that it is possible to orchestrate the learning procedure so that, for each information set, employing one regret minimizer per round does not compromise the overall convergence of the algorithm.
The empirical frequency of play generated by \icfr~converges to an EFCE almost surely in the limit.
These results generalize the seminal work by~\citet{hart2000simple} to the sequential case via a simple and natural framework.

{\color{violet}
\paragraph{Concurrent and Subsequent Work (Updated 2022)}

An updated and improved journal version of this paper is available on arXiv at \url{https://arxiv.org/abs/2104.01520} \cite{farina2021simple}. 
In the journal version, we completely revised the way in which our result is presented, by casting it into the framework of phi-regret minimization \citep{greenwald2003general,stoltz2007learning,gordon2008no}. This is a powerful improvement over the previous way of presenting our work. It helps to better connect the work to prior results that are also based on the phi-regret minimization framework.
Notably, we also strengthened our results by providing high-probability convergence bounds for EFCE that hold at finite time, besides almost-sure convergence results in the limit. The conference version only included an almost-sure convergence guarantee. 

In this section we mention related work that surfaced after the publication of the present paper, and summarize some recent trends related to phi-regret dynamics in games.

First, we acknowledge the thesis work by \citet{zhang2022simple} on computing certain refinements of EFCE via polynomial-time uncoupled learning dynamics, though their procedure could require up to exponential memory. In a later chapter of the thesis, the author notes that some of the dynamics introduced in the thesis can be modified to guarantee polynomial memory usage when convergence to the set of (unrefined) EFCE is sought. That work was conducted independently and concurrently with ours.

In a subsequent paper, \citet{morrill2021efficient} extend some of our results by conducting a study of different forms of correlation in extensive-form games, defining a taxonomy of solution concepts that includes in particular EFCE.
Each of their solution concepts is attained by a particular set of no-regret learning dynamics, which is obtained by instantiating the phi-regret minimization framework with a suitably-defined deviation function.
%
% As part of their analysis, \citet{morrill2020hindsight} investigate some properties of the well-established CFR regret minimization algorithm~\citep{zinkevich2008regret} applied to $n$-player general-sum extensive-form games, establishing that it is hindsight-rational with respect to a specific set of deviation functions, which the authors coin {\em blind counterfactual deviations}.
%
Specifically, they identify a general class of deviations---called {\em behavioral deviations}---that induce equilibria that can be found through uncoupled no-regret learning dynamics. Behavioral deviations are defined as those specifying an action transformation independently at each information set of the game. As the authors note, the deviation functions involved in the definition of EFCE do not fall under that category. A particular class of behavioral deviation functions---called \emph{causal partial sequence deviations}---induces solution concepts that are (subsets of) EFCEs. So, their result begets an alternative set of no-regret learning dynamics that converge to EFCE, based on a different set of deviation functions than those we use in this article.

A somewhat recent trend in the learning in game literature has seen the introduction of the concept of \emph{optimistic} learning dynamics which can guarantee convergence to the set of equilibria faster than the $O(1/\sqrt{T})$ rate attainable in the fully adversarial setting.
That line of work was pioneered by~\citet{Daskalakis11:Near}, and has since been extended along several lines~\citep{Rakhlin13:Online,Rakhlin13:Optimization,Syrgkanis15:Fast,Chen20:Hedging,Daskalakis21:Near,Daskalakis21:Fast,Piliouras21:Optimal}, incorporating partial or noisy information feedback~\citep{Foster16:Learning,Wei18:More,Hsieh22:No}, and more recently, general Markov games~\citep{Erez22:Regret,Zhang22:Policy}.
In the case of EFCE dynamics, the work of~\citet{Anagnostides22:Faster} establishes $O(T^{1/4})$ trigger regret bounds through \emph{optimistic hedge}, building on~\citep{Chen20:Hedging}, they showed \emph{multiplicative stability} of the fixed points associated with EFCE.

Finally, we acknowledge some very recent papers that have developed dynamics converging to EFCE under bandit feedback~\citep{Bai22:Efficient,Song22:Sample}.
}
    \section{Preliminaries}

In this section, we provide some groundings on sequential games and regret minimization (see the books by~\citet{shoham2008multiagent} and~\citet{cesa2006prediction}, for additional details).

\subsection{Extensive-form games}

We focus on \emph{extensive-form games} (EFGs) with imperfect information.
We denote the set of players as $\pl \cup\{c\}$, where $c$ is a \emph{chance player} that selects actions according to fixed known probability distributions, representing exogenous stochasticity.
An EFG is usually defined by means of a \emph{game tree}, where $H$ is the set of nodes of the tree, and a node $h \in H$ is identified by the ordered sequence of actions from the root to the node.
$Z \subseteq H$ is the set of terminal nodes, which are the leaves of the tree.
For every $h \in H \setminus Z$, we let $P(h) \in \pl \cup \{c\}$ be the unique player who acts at $h$ and $A(h)$ be the set of actions she has available.
%
% We write $h \cdot a$ to denote the node reached when $a \in A(h)$ is played at $h$.
%
For each player $i \in \mathcal{P}$, we let $u_i: Z \rightarrow \mathbb{R}$ be her payoff function.
Moreover, we denote by $p_c: Z \to (0,1)$ the function assigning each terminal node $z \in Z$ to the product of probabilities of chance moves encountered on the path from the root of the game tree to $z$.

Imperfect information is encoded by using \emph{information sets} (infosets).
Given $i \in \pl$, a player $i$'s infoset $I$ groups nodes belonging to player $i$ that are indistinguishable for her, \emph{i.e.}, $A(h) = A(k)$ for any pair of nodes $h, k \in I$.
$\mathcal{I}_i$ denotes the set of all player $i$'s infosets.
Moreover, we let $A(I)$ be the set of actions available at infoset $I \in \I_i$.
As customary, we assume that the game has \emph{perfect recall}, \emph{i.e.}, the infosets are such that no player forgets information once acquired.
In EFGs with perfect recall, the infosets $\I_i$ of each player $i \in \pl$ are partially ordered.
We write $I \preceq J$ whenever infoset $I\in \I_i$ \emph{precedes} $J \in \I_i$ according to such ordering, \emph{i.e.}, formally, there exists a path in the game tree connecting a node $h \in I$ to some node $k \in J$.
For the ease of notation, given $I \in \I_i$, we let $\mathcal{C}^\star(I)$ be the set of player $i$'s infosets that follow infoset $I$ (this included), defined as $\mathcal{C}^\star(I) \coloneqq \{ J \in \I_i \mid I \preceq J \}$.
Moreover, given $I \in \I_i$ and $a \in A(I)$, we let $\mathcal{C}(I,a) \subseteq \I_i$ be the set of player $i$'s infosets that immediately follow $I$ by playing action $a$, \emph{i.e.}, those reachable from at least one node $h \in I$ by following a path that includes $a$ and does not pass through another infoset of $i$.

\paragraph{Normal-form plans and strategies}
A \textit{normal-form plan} for player $i \in \pl$ is a tuple $\pi_i\in\Pi_i  \coloneqq \bigtimes_{I\in\mathcal{I}_i} A(I)$ which specifies an action for each player $i$'s infoset, where $\pi_i(I)$ represents the action selected by $\pi_i$ at infoset $I\in \mathcal{I}_i$.
We denote with $\pi \in \Pi \coloneqq \bigtimes_{i\in \pl}\Pi_i$ a \emph{joint normal-form plan}, defining a plan $\pi_i \in \Pi_i$ for each player $i \in \pl$.
Moreover, a tuple defining normal-form plans for the opponents of player $i \in \pl$ is denoted as $\pi_{-i} \in \Pi_{-i} \coloneqq \bigtimes_{j \neq i \in \mathcal{P}} \Pi_j$.
%
%The expected payoff of player $i \in \pl$, when she plays $\pi_i \in \Pi_i$ and the opponents play normal-form plans specified by $\pi_{-i} \in \Pi_{-i} \coloneqq \bigtimes_{j \neq i \in \mathcal{P}} \Pi_j$, is denoted, with an overload of notation, by $u_i(\pi_i,\pi_{-i})$ (this also includes the probability of chance moves, as determined by $p_c$).
%
A \textit{normal-form strategy} $\mu_i \in \Delta_{\Pi_i}$ is a probability distribution over $\Pi_i$, where $\mu_i [\pi_i]$ denotes the probability of selecting a plan $\pi_i \in \Pi_i$ according to $\mu_i$.
Moreover, $\mu \in \Delta_{\Pi}$ is a \emph{joint probability distribution} defined over $\Pi$, with $\mu [\pi]$ being the probability that the players end up playing the plans prescribed by $\pi \in \Pi$.

\paragraph{Sequences}
For any player $i \in \pl$, given an infoset $I \in \I_i$ and an action $a \in A(I)$, we denote with $\sigma = (I,a)$ the \emph{sequence} of player $i$'s actions reaching infoset $I$ and terminating with $a$.
Notice that, in EFGs with perfect recall, such sequence is uniquely determined, as paths that reach nodes belonging to the same infoset identify the same sequence of player $i$'s actions.
We let $\Sigma_i \coloneqq \{ (I,a) \mid I \in \I_i , a \in A(I) \} \cup \{ \varnothing_i \}$ be the set of player $i$'s sequences, where $\varnothing_i$ is the empty sequence of player $i$ (representing the case in which she never plays).
Additionally, given an infoset $I \in \I_i$, we let $\sigma(I) \in \Sigma_i$ be the sequence of player $i$'s actions that identify infoset $I$.

\paragraph{Subsets of (joint) normal-form plans}
We now define a few useful subsets of $\Pi_i$. The reader is encouraged to refer to Figure~\ref{fig:notation} for a simple example.
For every player $i \in \pl$ and infoset $I \in \I_i$, we let $\Pi_{i}(I) \subseteq \Pi_i$ be the set of player $i$'s normal-form plans that prescribe to play so as to reach infoset $I$ whenever possible (depending on the opponents' actions up to that point) and \emph{any} action whenever reaching $I$ is \emph{not} possible anymore.
Moreover, for every sequence $\sigma = (I,a) \in \Sigma_i$, we let $\Pi_{i}(\sigma) \subseteq \Pi_{i}(I) \subseteq \Pi_i$ be the set of player $i$'s plans that reach infoset $I$ and recommend action $a$ at $I$.
Similarly, given a terminal node $z \in Z$, we denote with $\Pi_i(z) \subseteq \Pi_i$ the set of normal-form plans by which player $i$ plays so as to reach $z$, while $\Pi(z) \coloneqq \bigtimes_{i \in \pl} \Pi_i(z)$ and $\Pi_{-i}(z) \coloneqq \bigtimes_{j \neq i \in \pl} \Pi_j(z)$.
%
% Finally, for convenience, we let $\Pi_i(\sigma) \coloneqq \Pi_i(I,a)$ for any sequence $\sigma = (I,a) \in \Sigma_i$.
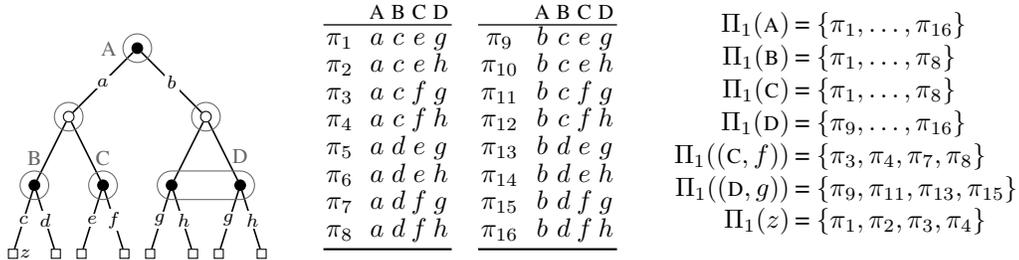
\begin{figure}[H]\centering
	{\hspace{-0.3cm}
		\begin{minipage}[b]{3.5cm}\centering%
			\def\done{.8*1.2}
			\def\dtwo{.40*1.2}
			\def\dleaf{.25*1.2}
			\def\dvert{.8*1.2}
			\begin{tikzpicture}[baseline=-1.1cm,scale=.95]
			\node[fill=black,draw=black,circle,inner sep=.5mm] (A) at (0, 0) {};
			\node[fill=white,draw=black,circle,inner sep=.5mm] (X) at ($(-\done,-\dvert)$) {};
			\node[fill=white,draw=black,circle,inner sep=.5mm] (Y) at ($(\done,-\dvert)$) {};
			\node[fill=black,draw=black,circle,inner sep=.5mm] (B) at ($(X) + (-\dtwo, -\dvert)$) {};
			\node[fill=black,draw=black,circle,inner sep=.5mm] (C) at ($(X) + (\dtwo, -\dvert)$) {};
			\node[fill=white,draw=black,inner sep=.6mm] (l1) at ($(B) + (-\dleaf, -\dvert)$) {};
			\node[fill=white,draw=black,inner sep=.6mm] (l2) at ($(B) + (\dleaf, -\dvert)$) {};
			\node[fill=white,draw=black,inner sep=.6mm] (l3) at ($(C) + (-\dleaf, -\dvert)$) {};
			\node[fill=white,draw=black,inner sep=.6mm] (l4) at ($(C) + (\dleaf, -\dvert)$) {};
			
			\node[inner sep=0] at ($(l1) + (.17,0)$) {\scriptsize$z$};
			
			\node[fill=black,draw=black,circle,inner sep=.5mm] (D1) at ($(Y) + (-\dtwo, -\dvert)$) {};
			\node[fill=black,draw=black,circle,inner sep=.5mm] (D2) at ($(Y) + (\dtwo, -\dvert)$) {};
			\node[fill=white,draw=black,inner sep=.6mm] (l5) at ($(D1) + (-\dleaf, -\dvert)$) {};
			%\node[fill=white,draw=black,inner sep=.6mm] (l6) at ($(D1) + (0, -\dvert)$) {};
			\node[fill=white,draw=black,inner sep=.6mm] (l7) at ($(D1) + (\dleaf, -\dvert)$) {};
			\node[fill=white,draw=black,inner sep=.6mm] (l8) at ($(D2) + (-\dleaf, -\dvert)$) {};
			%\node[fill=white,draw=black,inner sep=.6mm] (l9) at ($(D2) + (0, -\dvert)$) {};
			\node[fill=white,draw=black,inner sep=.6mm] (l10) at ($(D2) + (\dleaf, -\dvert)$) {};

			\draw[semithick] (A) --node[fill=white,inner sep=.9] {\scriptsize$a$} (X);
			\draw[semithick] (A) --node[fill=white,inner sep=.9] {\scriptsize$b$} (Y);
			\draw[semithick] (B) --node[fill=white,inner sep=.9] {\scriptsize$c$} (l1);
			\draw[semithick] (B) --node[fill=white,inner sep=.9] {\scriptsize$d$} (l2);
			\draw[semithick] (C) --node[fill=white,inner sep=.9] {\scriptsize$e$} (l3);
			\draw[semithick] (C) --node[fill=white,inner sep=.9] {\scriptsize$f$} (l4);
			%\draw[semithick] (D1) --node[fill=white,inner xsep=0,rounded corners=.5mm,inner ysep=.9,yshift=.4pt] {\scriptsize$h$} (l6);
			\draw[semithick] (D1) --node[fill=white,inner xsep=0,inner ysep=.9,xshift=-.4] {\scriptsize$g$} (l5);
			\draw[semithick] (D1) --node[fill=white,inner xsep=.2mm,inner ysep=.9,xshift=.6] {\scriptsize$h$} (l7);
			%\draw[semithick] (D2) --node[fill=white,inner xsep=0,rounded corners=.5mm,inner ysep=.9,yshift=.4pt] {\scriptsize$h$} (l9);
			\draw[semithick] (D2) --node[fill=white,inner xsep=0,inner ysep=.9,xshift=-.4] {\scriptsize$g$} (l8);
			\draw[semithick] (D2) --node[fill=white,inner xsep=.2mm,inner ysep=.9,xshift=.6] {\scriptsize$h$} (l10);
			\draw[semithick] (B) -- (X) -- (C);
			\draw[semithick] (D1) -- (Y) -- (D2);
			
			\draw[black!60!white] (X) circle (.2);
			\draw[black!60!white] (Y) circle (.2);
			
			\draw[black!60!white] (A) circle (.2);
			\node[black!60!white]  at ($(A) + (-.4, 0)$) {\textsc{a}};
			
			\draw[black!60!white] (B) circle (.2);
			\node[black!60!white]  at ($(B) + (0,.38)$) {\textsc{b}};
			
			\draw[black!60!white] (C) circle (.2);
			\node[black!60!white]  at ($(C) + (0,.38)$) {\textsc{c}};
			
			\draw[black!60!white] ($(D1) + (0, .2)$) arc (90:270:.2);
			\draw[black!60!white] ($(D1) + (0, .2)$) -- ($(D2) + (0, .2)$);
			\draw[black!60!white] ($(D1) + (0, -.2)$) -- ($(D2) + (0, -.2)$);
			\draw[black!60!white] ($(D2) + (0, -.2)$) arc (-90:90:.2);
			\node[black!60!white]  at ($(D2) + (0, .4)$) {\textsc{d}};
			\end{tikzpicture}
	\end{minipage}}
	\hspace{.5cm}
	{\begin{minipage}[b]{4cm}\centering
			\setlength{\tabcolsep}{1pt}
			\begin{tabular}{c@{\hskip 6pt}cccc}
				& \textsc{a} & \textsc{b} & \textsc{c} & \textsc{d} \\[-1mm]
				\midrule\\[-5mm]
				$\pi_1$ & $a$ & $c$ & $e$ & $g$  \\[-.2mm]
				$\pi_2$ & $a$ & $c$ & $e$ & $h$  \\[-.2mm]
				$\pi_3$ & $a$ & $c$ & $f$ & $g$  \\[-.2mm]
				$\pi_4$ & $a$ & $c$ & $f$ & $h$  \\[-.2mm]
				$\pi_5$ & $a$ & $d$ & $e$ & $g$  \\[-.2mm]
				$\pi_6$ & $a$ & $d$ & $e$ & $h$  \\[-.2mm]
				$\pi_7$ & $a$ & $d$ & $f$ & $g$  \\[-.2mm]
				$\pi_8$ & $a$ & $d$ & $f$ & $h$  \\[-.2mm]
				\bottomrule
			\end{tabular}\quad
			\begin{tabular}{c@{\hskip 6pt}cccc}
				& \textsc{a} & \textsc{b} & \textsc{c} & \textsc{d} \\[-1mm]
				\midrule\\[-5mm]
				$\pi_9$    & $b$ & $c$ & $e$ & $g$  \\[-.2mm]
				$\pi_{10}$ & $b$ & $c$ & $e$ & $h$  \\[-.2mm]
				$\pi_{11}$ & $b$ & $c$ & $f$ & $g$  \\[-.2mm]
				$\pi_{12}$ & $b$ & $c$ & $f$ & $h$  \\[-.2mm]
				$\pi_{13}$ & $b$ & $d$ & $e$ & $g$  \\[-.2mm]
				$\pi_{14}$ & $b$ & $d$ & $e$ & $h$  \\[-.2mm]
				$\pi_{15}$ & $b$ & $d$ & $f$ & $g$  \\[-.2mm]
				$\pi_{16}$ & $b$ & $d$ & $f$ & $h$  \\[-.2mm]
				\bottomrule
			\end{tabular}
	\end{minipage}}
	\hspace{.5cm}
	{\begin{minipage}[b]{4cm}
			\setlength{\tabcolsep}{1pt}
			\begin{tabular}{rl}
				$\Pi_1(\textsc{a})$ &= $\{\pi_1,\dots,\pi_{16}\}$\\[.5mm]
				$\Pi_1(\textsc{b})$ &= $\{\pi_1, \dots, \pi_8\}$\\[.5mm]
				$\Pi_1(\textsc{c})$ &= $\{\pi_1, \dots, \pi_8\}$\\[.5mm]
				$\Pi_1(\textsc{d})$ &= $ \{\pi_9, \dots, \pi_{16}\}$\\[.5mm]
				$\Pi_1((\textsc{c}, f))$ &= $ \{\pi_3, \pi_4, \pi_7, \pi_8\}$\\[.5mm]
				$\Pi_1((\textsc{d}, g))$ &= $ \{\pi_9, \pi_{11}, \pi_{13}, \pi_{15}\}$\\[.5mm]
				$\Pi_1(z)$ &= $\{\pi_1, \pi_2, \pi_3, \pi_4\}$
			\end{tabular}%
	\end{minipage}}
	\caption{(Left) Sample game tree. Black round nodes belong to Player $1$, white round nodes belong to Player $2$, and white square nodes are leaves. Rounded, gray lines denote information sets. (Center) Set $\Pi_1$ of normal-form plans for Player $1$. Each plan identifies an action at each information set. (Right) Examples of certain subsets of $\Pi_1$ defined in this subsection.}\label{fig:notation}
\end{figure}

\paragraph{Additional notation}
For every $i \in \pl$ and $I \in \I_i$, we let $Z(I) \subseteq Z$ be the set of terminal nodes that are reachable from infoset $I \in \I_i$ of player $i$.
%
%, while, for the ease of notation, $\X_{i}(I) \coloneqq \Delta_{\Pi_{i}(I)}$ denotes the set of probability distributions defined over plans in $\Pi_{i}(I)$.
%
Moreover, $Z(I, a) \subseteq Z(I) \subseteq Z$ is the set of terminal nodes reachable by playing action $a \in A(I)$ at infoset $I $, whereas $Z^c(I,a) \coloneqq Z(I) \setminus Z(I,a)$ is the set of terminal nodes which are reachable by playing an action different from $a$ at $I$.
For any player $i\in\pl$, normal-form plan $\pi_i\in\Pi_i$, infoset $I\in\I_i$, and terminal node $z\in Z$, we define $\rho^{\pi_i}_{I\to z}$ as a function equal to $1$ if $z$ is reachable from $I$ when player $i$ plays according to $\pi_i$, and $0$ otherwise.
Finally, we define a notion of {\em reach} such that, for each normal-form plan $\pi=(\pi_i,\pi_{-i})\in\Pi$, infoset $I\in \I_i$, and terminal node $z\in Z$, we have $\rho^{(\pi_i,\pi_{-i})}_{I\to z}\coloneqq \rho^{\pi_i}_{I\to z} \cdot 
\one[\pi_{-i}\in\Pi_{-i}(z)]$.

\subsection{External and internal regret minimization}

%A regret minimizer is a device that supports two operations:  (i)RECOMMEND, which providesthe next decisionxt+1∈ X, whereXis a nonempty, convex, and compact subset of a EuclideanspaceRn; and (ii)OBSERVELOSS, which receives/observes a convex loss function`tthat is usedto evaluate decisionxt[Zinkevich, 2003].  In this paper, we will consider linear loss functions,which we represent in the form of a vector`t∈Rn.  A regret minimizer is anonlinedecisionmaker in the sense that each decision is made by taking into account only past decisions and theircorresponding losses. The quality metric for the regret minimizer is itscumulative regretRT, definedas the difference between the loss cumulated by the sequence of decisionsx1,...,xTand the lossthat would have been cumulated by thebest-in-hindsight time-independentdecisionˆx.  Formally,RT:=∑Tt=1〈`t,xt〉−minˆx∈X∑Tt=1〈`t,ˆx〉.A ‘good’ regret minimizer hasRTsublinear inT;this property is known asHannan consistency.  

In the {regret minimization framework}~\citep{zinkevich2003online}, each player $i \in \pl$ plays repeatedly against the others by making a series of decisions from a set $\X_i$.
A \emph{regret minimizer} for player $i \in \pl$ is a device that, at each iteration $t = 1, \ldots, T$, supports two operations: (i) \textsc{Recommend}, which provides the next decision $x_i^{t+1} \in \X_i$ on the basis of the past history of play and the observed utilities up to iteration $t$; and (ii) \textsc{Observe}, which receives a utility function $u_i^t : \X_i \to \mathbb{R}$ that is used to evaluate decision $x_i^t$.
A regret minimizer is evaluated in terms of its cumulative regret.
Two types of regret minimizers are commonly studied, depending on the adopted notion of regret, either \emph{external} or \emph{internal} regret. 
%
%For each $t = 1, \ldots, T$, let $\pi_i^t \in \Pi_i$ be the normal-form plan adopted by player $i$ at iteration $t$.
%%
%Then, each player $i$ observes a utility defined as $u_i^t: \Pi_i \to \mathbb{R}$ with $u_i^t(\pi_i) \coloneqq u_i(\pi_i, \pi_{-i}^t)$ for $\pi_i \in \Pi_i$, where $\pi_{-i}^t \in \Pi_{-i}$ collectively denotes the normal-form plans played by player $i$'s opponents at iteration $t$.
%\todo{questa u non serve}

\paragraph{External regret}
An \emph{external-regret minimizer} $\rme$ for player $i \in \pl$ is a device minimizing the \emph{cumulative external regret} of player $i$ up to iteration $T$, which is defined as:
\begin{equation}\label{eq:ext_regret}
R_i^T \coloneqq \max_{\hat x_i\in \X_i}\mleft\{\sum_{t=1}^T u_i^t(\hat x_i)\mright\}-\sum_{t=1}^Tu_i^t(x_i^t).
\end{equation}
$R_i^T$ represents how much player $i$ would have gained by always taking the best decision in hindsight, given the history of utilities observed up to iteration $T$.
%
% We denote an \emph{external-regret minimizer} with $\rme$.
%
%A desirable property for an $\rme$ is \emph{Hannan consistency}~\citep{hannan1957approximation}, which requires that  $\limsup_{T\to\infty} \frac{1}{T} R_i^T\leq 0$.
%%
%This property is achieved, \emph{e.g.}, by \emph{regret matching} (RM)~\citep{hart2000simple}.
%
%A \emph{external-regret minimizer} $\rme$ is a device providing, after each iteration $t$, the next player $i$'s normal-form plan $\pi_i^{t+1}$ on the basis of the past history of play and the observed utilities up to iteration $t$.
%%
%A desirable property for an external-regret minimizer is \emph{Hannan consistency}~\citep{hannan1957approximation}, which requires that  $\limsup_{T\to\infty} \frac{1}{T} R_i^T\leq 0$, \emph{i.e.}, the cumulative regret grows at a sublinear rate in the number of iterations $T$.
%%
%There are many regret-minimizing procedures that ensure such property, one is \emph{regret matching} (RM)~\citep{hart2000simple}.

\paragraph{Internal regret}
An \emph{internal-regret minimizer} $\rmi$ for player $i \in \pl$ is a device minimizing the {\em cumulative internal regret} of player $i$ up to iteration $T$, which is defined as:
\begin{equation}\label{eq:int_regret}
\max_{x_i,\hat x_i\in \X_i} R_{i,(x_i,\hat x_i)}^T\coloneqq \max_{x_i,\hat x_i\in \X_i}\mleft\{\sum_{t=1}^T\one[x_i=x_i^t]\mleft( u_i^t(\hat x_i) -u_i^t( x_i)\mright) \mright\}.
\end{equation}
Intuitively, player $i$ has small internal regret if, for each pair of decisions $(x_i,\hat x_i)$, she does not regret of not having played $\hat x_i$ each time she selected $x_i$.
The notion of internal regret is strictly stronger than the notion of external regret: any algorithm with small internal regret also has small external regret, but the converse does not hold (see~\citet{stoltz2005internal} for an example).

Regret minimizers show an interesting connection with games when the decision sets $\X_i$ are the sets of normal-form plans $\Pi_i$ and the observed utilities $u_i^t$ are obtained by playing the game according to the selected plans $\pi_i^t$.
Letting $\pi^t \coloneqq ( \pi_i^t )_{i \in \pl}$ be the joint normal-form plan resulting at each iteration $t = 1,\ldots, T$, we denote with $\{ \pi^t \}_{t =1}^T$ the overall sequence of plays made by the players.
Then, the \emph{empirical frequency of play} $\bar{\mu}^T\in\Delta_\Pi$ generated by $\{ \pi^t \}_{t =1}^T$ is such that for every $\pi \in \Pi$:
\begin{equation}\label{eq:empirical frequency}
\bar \mu^T(\pi) \coloneqq \frac{|\{1\leq t \leq T \mid \pi^t = \pi  \}|}{T}.
\end{equation}
If all the players play according to some external-regret minimizers, then $\bar \mu^T$ approaches the set of (normal-form) coarse correlated equilibria, even in EFGs (see~\citet{cesa2006prediction} and~\citet{celli2019learning} for further details).
%
%an Hannan consistent regret-minimizing procedure, then $\bar \mu^T$ approaches the set of (normal-form) coarse correlated equilibria even in sequential games (see~\citet{cesa2006prediction} and~\citet{celli2019learning} for further details).
%
Moreover, \citet{foster1997calibrated} and \citet{hart2000simple} established that the empirical frequency of play generated by any no-internal-regret algorithm (see~\citet{cesa2006prediction} and~\citet{blum2007external} for some examples) converges to the set of correlated equilibria in repeated games with simultaneous moves ({\em i.e.,} normal-form games).
%
% We denote an arbitrary {\em internal-regret minimizer} by $\rmi$ . 
    \section{Extensive-form correlated equilibria}\label{sec:equilibria}

The definition of EFCE requires the following notion of trigger agent, which, intuitively, is associated to each player and each of her sequences of action recommendations.

\begin{definition}[Trigger agent for EFCE]
	Given a player $i \in \pl$, a sequence $\sigma = (I,a) \in \Sigma_i$
	%
	%, an infoset $I \in \I_i$, an action $a \in A(I)$
	%
	, and a probability distribution $\hat \mu_i \in \Delta_{\Pi_i(I)}$, an \emph{$(\sigma, \hat \mu_i)$-trigger agent for player $i$} is an agent that takes on the role of player $i$ and commits to following all recommendations unless she reaches $I$ and gets recommended to play $a$.
	If this happens, the player stops committing to the recommendations and plays according to a plan sampled from $\hat \mu_i$ until the game ends.
\end{definition}

It follows that joint probability distribution $\mu \in \Delta_\Pi$ is an EFCE if, for every $i \in \pl$, player $i$'s expected utility when following the recommendations is at least as large as the expected utility that any $(\sigma, \hat \mu_i)$-trigger agent for player $i$ can achieve (assuming the opponents' do not deviate).

For any $\mu\in\Delta_\Pi$, sequence $\sigma = (I,a) \in \Sigma_i$, and $(\sigma, \hat \mu_i)$-trigger agent, we define the probability of the game ending in a terminal node $z \in Z(I)$ as:
\begin{equation}
p_{\mu,\hat\mu_i}^\sigma(z) \coloneqq \left( \sum_{\substack{\pi_i \in \Pi_i(\sigma)\\\pi_{-i} \in \Pi_{-i}(z)}} \mu (\pi_i, \pi_{-i}) \right) \left( \sum_{\hat \pi_i \in \Pi_i(z)} \hat \mu_i(\hat \pi_i)  \right) p_c(z),
\end{equation}
which accounts for the fact that the agent follows recommendations until she receives the recommendation of playing $a$ at $I$, and, thus, she `gets triggered' and plays according to $\hat \pi_i$ sampled from $\hat \mu_i$ from $I$ onwards.
Moreover, the probability of reaching a terminal node  $z \in Z(I,a)$ when following the recommendations is defined as follows:
\begin{equation}\label{eq:q}
q_\mu(z) \coloneqq \left( \sum_{\pi \in \Pi(z)} \mu(\pi) \right) p_c(z).
\end{equation}
The definition of EFCE reads as follows (see Appendix~\ref{appendix:efce} or the work by~\citet{farina2019correlation} for details):
\begin{definition}[Extensive-form correlated equilibrium]\label{def:efce}
	An EFCE of an EFG is a joint probability distribution $\mu \in \Delta_\Pi$ such that, for every $i \in  \pl$ and $(\sigma, \hat \mu_i)$-trigger agent for player $i$, with $\sigma =(I,a) \in \Sigma_i$, it holds:
	\begin{equation}\label{eq:efce_simp}
	\sum_{z \in Z(I,a)} q_\mu(z) u_i(z) \geq \sum_{z \in Z(I)} p_{\mu,\hat\mu_i}^\sigma(z) u_i(z).
	\end{equation}
\end{definition}

A joint probability distribution $\mu\in\Delta_\Pi$ is said to be an $\epsilon$-EFCE when the {\em maximum deviation} $\delta(\mu)$ under $\mu$ is such that: 
\begin{equation}\label{eq:efce_simp_eps}
\delta(\mu)\defeq \max_{i\in\pl}\max_{\sigma=(I,a)\in\Sigma_i} \mleft\{ \max_{\hat\mu_i \in \Delta_{\Pi_i(I)}} \mleft\{ \sum_{z\in Z(I)} p^{\sigma}_{\mu,\,\hat\mu_i}(z)  u_i(z)\mright\} - \sum_{z\in Z(I,a)} q_{\mu}(z) u_i(z)\mright\}
\le \epsilon.
\end{equation}

    \section{Trigger regret and relationships with EFCE}\label{sec:efce_regret}

In this section, we introduce the notion of \emph{trigger regret}. Intuitively, it measures the regret that each trigger agent has for not having played the best-in-hindsight strategy. As we will show, when each trigger agent has low trigger regret, then the empirical frequency of play is close to being an EFCE.

%Our main idea is to define a regret for each trigger agent, \emph{i.e.}, for each player $i \in \pl$ and sequence $\sigma =(I,a) \in \Sigma_i$.
%
%Intuitively, this represents the regret of not having played the best trigger agent's plan $\hat \pi_i \in \Pi_i(I)$ in hindsight, taking into account all the iterations in which the agent actually gets triggered (\emph{i.e.}, infoset $I$ is reached and action $a$ is recommended).

Given a sequence $\{ \pi^t \}_{t=1}^T$, the vector of {\em immediate utilities} ${u}_i^t$ observed by player $i \in \pl$ after any iteration $t = 1,\ldots, T$ is defined as follows.
For every infoset $I \in \I_i$ and action $a \in A(I)$ we have:
\[
u_i^t[I,a] \coloneqq \sum_{z \in Z(I,a) \setminus \bigcup_{J \in \mathcal{C}(I,a) } Z(J) }\one[\pi_{-i}^t\in\Pi_{-i}(z)]\, p_c(z) u_i(z) ,
\]
which represents the utility experienced by player $i$ if the game ends after playing action $a$ at infoset $I$, without going through other player $i$'s infosets and assuming that the other players play as prescribed by the plans $\pi_{-i}^t\in\Pi_{-i}$ at iteration $t$.
Notice that the summation is over the terminal nodes immediately reachable from $I$ by playing $a$ and the payoff of each terminal node is multiplied by the probability of reaching it given chance probabilities.

For $i \in \pl$, the following recursive formula defines player $i$'s utility attainable at infoset $I\in\I_i$ when a normal-form plan $\pi_i \in \Pi_{i}$ is selected:
\begin{equation}\label{eq:v_definition}
V_I^t( \pi_{i}) \coloneqq u^t_{i}[I,\pi_{i}(I)] + \sum_{J \in\mathcal{C}({I,\pi_{i}(I)})} V_{J}^t( \pi_{i }).
\end{equation}

\begin{definition}[Trigger regret]\label{def:trigger_regret}
	For every player $i\in\pl$ and sequence $\sigma=(I,a)\in\Sigma_i$,  we let $R^T_\sigma$ be the {\em trigger regret} for sequence $\sigma$, which we define as follows:
	\begin{equation*}
	R^T_{\sigma}\coloneqq \max_{\hat\pi_i\in\Pi_i(I)} \mleft\{\sum_{t=1}^T
	\one[\pi^t_i\in\Pi_i(\sigma)]\,\Big(V_I^t(\hat\pi_i)- V^t_I(\pi_i^t)\Big)\mright\}.
	\end{equation*}
\end{definition}
The trigger regret for $\sigma=(I,a)$ represents the regret experienced by the trigger agent that gets triggered on sequence $\sigma$, \emph{i.e.}, when infoset $I$ is reached and action $a$ is recommended.
Notice that $R^T_{\sigma}$ only accounts for those iterations in which $\pi^t_i\in\Pi_i(\sigma)$, \emph{i.e.}, intuitively, when the actions prescribed by the normal-form plan $\pi_i^t$ trigger the agent associated to sequence $\sigma$.

The following theorem shows that minimizing the trigger regrets for each player $i\in\pl$ and sequence $\sigma\in\Sigma_i$ allows to approach the set of EFCEs.
\begin{restatable}{theorem}{efceTh}\label{th:efce}
	At all times $T$, the empirical frequency of play $\bar \mu^T$ (Equation~\ref{eq:empirical frequency}) is an $\epsilon$-EFCE, where
    \[
        \epsilon \defeq \max_{i\in\pl}\max_{\sigma\in\Sigma_i} \frac{R^T_{\sigma}}{T}.
    \]
\end{restatable}

\begin{restatable}{corollary}{corefceconvergence}\label{cor:efce convergence}
    If $\displaystyle\limsup_{T\to\infty} \max_{i\in\pl}\max_{\sigma\in\Sigma_i}\frac{R^T_\sigma}{T} \le 0,
    $
    then $\displaystyle\limsup_{T\to\infty} \delta(\bar\mu^T) \le 0$, that is, for any $\epsilon > 0$, eventually the empirical frequency of play $\bar{\mu}^T$ becomes an $\epsilon$-EFCE.
\end{restatable}

    \section{Laminar regret decomposition for trigger regret}\label{sec:laminar}

In order to design an algorithm minimizing trigger regrets, we first develop a new regret decomposition that extends the {\em laminar regret decomposition} framework introduced by~\citet{farina2018online}.
Our decomposition exploits the structure of the EFG to show that trigger regrets can be minimized by minimizing other suitably defined regret terms which are {\em local} at each infoset.

First, for each player $i\in\pl$, sequence $\sigma=(J,a)\in\Sigma_i$, and infoset $I \in\mathcal{C}^\star(J)$ ({\em i.e.}, any infoset following from $J$, this included),
we define the notion of {\em subtree regret} as follows:
\begin{equation*}
	R_{\sigma,I}^T\coloneqq \max_{\hat \pi_i\in\Pi_i(I)}\mleft\{ \sum_{t=1}^T \one[\pi_i^t\in\Pi_i(\sigma)] \, \Big( V_I^t(\hat \pi_i) - V^t_I(\pi^t_i) \Big)  \mright\} .
\end{equation*}
Each term $R_{\sigma,I}^T$ represents the regret at infoset $I$ experienced by the trigger agent that gets triggered on sequence $\sigma = (J,a)$.
Differently from the trigger regret $R_{\sigma}^T$, which is defined only for the infoset $J$ of $\sigma$, the subtree regrets $R_{\sigma, I}^T$ are defined for all the infosets $I \in \I_i$ such that $J \preceq I$.
%sequences $\sigma=(J,a)\in \Sigma_i$ such that $J\preceq I$.
%
\begin{remark}
Given player $i\in\pl$, it is immediate to see that, if $R_{\sigma,I}^T=o(T)$ for each $\sigma = (J,a)\in\Sigma_i$ and $I \in \mathcal{C}^\star(J)$, then $R_\sigma^T=o(T)$ for every $\sigma \in \Sigma_i$.
Therefore, we can safely focus on the problem of minimizing subtree regrets, as this will automatically guarantee convergence to an EFCE.
\end{remark}

Next, we need to introduce, for every player $i \in \pl$ and infoset $I \in \I_i$, the following parameterized utility function defined at each iteration $t = 1, \ldots, T$:
\begin{equation}\label{parametrized_util}
\hat u_{I}^t: A(I) \ni a \mapsto u_{i}^t[I,a] + \sum_{J\in\mathcal{C}(I,a)}  V_{J}^t(\pi^t_i),
\end{equation}
which represents the utility that player $i$ gets, at iteration $t$, by playing action $a$ at $I$ and following the actions prescribed by $\pi_i^t$ at the subsequent infosets.
Then, for each sequence $\sigma  =(J,a') \in \Sigma_i$, infoset $I \in \mathcal{C}^\star(J)$, and action $a \in A(I)$, the \emph{laminar subtree regret of action $a$} is defined as:
\begin{equation}\label{internal_laminar_action}
\hat R^T_{\sigma,I,a}\coloneqq \sum_{t=1}^T \one[\pi_i^t\in\Pi_i(\sigma)]\, \Big(\hat u_{I}^t(a) - \hat u_{I}^t(\pi_i^t(I))\Big),
\end{equation}
while, for $\sigma=(J,a')  \in \Sigma_i$ and $I \in \mathcal{C}^\star(J)$, the \emph{laminar subtree regret} is:
\begin{equation}\label{internal_laminar}
\hat R^T_{\sigma,I}\coloneqq \max_{a\in A(I)} \hat R_{\sigma,I,a}^T.
\end{equation}

The following two lemmas show that the subtree regrets can be minimized by minimizing the laminar subtree regrets at all the infosets of the game.

\begin{restatable}{lemma}{subtreeDecomposed}\label{lemma:cum_decomposed}
	The subtree regret for each player $i\in\pl$, sequence $\sigma=(J,a')\in \Sigma_i$, and infoset $I \in \mathcal{C}^\star(J)$ can be decomposed as:
	\begin{equation*}
	R_{\sigma,I}^T=\max_{ a\in A(I)}\mleft\{\hat R^T_{\sigma,I,a}+\sum_{I'\in\mathcal{C}(I, a)} R_{ \sigma,I'}^T\mright\}.
	\end{equation*}
\end{restatable}

The lemma is proved by recursively applying the  definitions of $R^T_{\sigma,I}$ and $V_I^t(\hat \pi_i)$, and by exploiting~\cref{parametrized_util}. Then,~\cref{lemma:cum_decomposed} is used to show the following.

\begin{restatable}{lemma}{laminarDecomposition}\label{lemma:regret_bound}
	For every player $i\in\pl$, sequence $\sigma =(J,a')\in \Sigma_i$, and infoset $I \in \mathcal{C}^\star(J)$, it holds:
	\begin{equation}
	R^T_{\sigma,I}\leq \max_{\hat\pi_i\in\Pi_{i}(I)}\sum_{I'\in \mathcal{C}^\star(I)} \one[\hat\pi_i \in \Pi_i(I')]\,\hat R_{\sigma,I'}^T.
	\end{equation}
\end{restatable}

    \section{Internal counterfactual regret minimization}\label{sec:icfr}

We propose the {\em internal counterfactual regret minimization} algorithm (\icfr) as a way to minimize the laminar subtree regrets described in the previous section.
At each iteration $t$, \icfr~builds a normal-form plan $\pi_i^t$ in a top-down fashion by sampling an action locally at each infoset, following a simple rule: if the current infoset can be reached through $\pi_i^t$, then an action is sampled according to an internal-regret minimizer; otherwise, an external-regret minimizer is employed.

\begin{wrapfigure}{r}{6.7cm}
	\hfill
	\scalebox{.9}{\begin{minipage}[t]{7.2cm}
			\begin{algorithm}[H]
				\small
				\caption{ICFR (for Player $i$)}
				\begin{algorithmic}[1]
					\Function{\textsc{ICFR}}{$i$}
					\State Initialize the regret minimizers
					\State $t\leftarrow 1$
					% \State $\pi_i^0\leftarrow $ Normal-form plan sampled from a uniform distribution at each $I\in\I_i$
					\While{$t<T$}
					\State $\pi_i^t\gets\textsc{SampleInternal}$\label{alg:sample}
					\State Observe ${u}_i^t$ ({\em i.e.}, $u_{i}^t[I,a]$ for each pair $(I,a)$\label{line:observe})
					\State $\textsc{UpdateInternal}(\pi_i^t,{u}_i^t)$
					\State $t\gets t+1$
					\EndWhile
					\EndFunction
					\Function{\textsc{SampleInternal}}{}
					\For{$I\in\I_i$ in a top-down order}
					\If{$\pi_i^t\in\Pi_i(I)$}
						\State $\pi_i^t(I)\gets \rmi_I.\textsc{Recommend}()$\label{line:internal}
					\Else
						\State $\sigma_I^t\gets \Sigma_i^c(I)\cap \{(J,\pi_i^t(J))\mid J\preceq I\}$\label{line:sigmaT}
						\State $\pi_i^t(I)\gets \rme_{\sigma_I^t,I}.\textsc{Recommend}()$\label{line:external}
					\EndIf
					\EndFor
					\EndFunction
					\Function{\textsc{UpdateInternal}}{$\pi_i^t,{u}_i^t$}
						\For{$I\in\I_i$}
						\State $\rmi_I.\textsc{Observe}(\one[\pi_i^t\in\Pi_i(I)]\cdot \hat u_I^t)$\label{line:up_int}
						\For{$\sigma\in\Sigma_i^c(I)$}
						\State $\rme_{\sigma,I}.\textsc{Observe}(\one[\pi_i^t\in\Pi_i(\sigma)]\cdot \hat u_I^t)$\label{line:up_ext}
						\EndFor
						\EndFor
					\EndFunction
				\end{algorithmic}%
				\label{alg:icfr}%
			\end{algorithm}
	\end{minipage}}
\end{wrapfigure}

In order to minimize the laminar subtree regrets, \icfr~needs to instantiate different regret minimizers for each infoset.
For every infoset $I\in \I_i$, the algorithm instantiates an internal-regret minimizer $\rmi_I$ employing an arbitrary no-internal-regret algorithm.
Moreover, let $\Sigma_i^c(I)\subseteq \Sigma_i$ be the set of sequences of player $i$ that do not allow to reach $I $ and whose last action is played at an infoset preceding $I$. Formally, 
\[
\Sigma_i^c(I) \coloneqq \{ (J, a)\in\Sigma_i \mid J\preceq I, a\notin \sigma(I) \}.
\]
\icfr~instantiates an additional external-regret minimizer $\rme_{\sigma,I}$ for each sequence $\sigma \in \Sigma_i^c(I)$.
The internal-regret minimizer $\rmi_I$ is responsible for the minimization of the laminar subtree regrets $\hat R_{\sigma, I}^T$ associated to trigger sequences $\sigma=(I,a)\in\Sigma_i$ for each $a\in A(I)$.
Instead, the external-regret minimizers $\rme_{\sigma,I}$ are responsible for the laminar subtree regrets of sequences $\sigma \in \Sigma_i^c(I) $.
%
%, and infoset $I$.
%
%Moreover, in order to minimize all the other subtree laminar regrets, \icfr~instantiates an additional external-regret minimizer $\rme_{\sigma,I}$ for each sequence $\sigma \in \Sigma_i^c(I)$, which is the set of sequences that do not allow to reach $I$ and whose last action is played at an infoset preceding $I$.
%
%That is,
%$
%\Sigma_i^c(I) \coloneqq \{ (J, a)\in\Sigma_i \mid J\preceq I, a\notin \sigma(I) \}.
%$

\cref{alg:icfr} provides a description of the procedures adopted by~\icfr.
At iteration $t$ and for each $I\in\I_i$, an action is sampled as follows: if the (possibly partial) normal-form plan $\pi_i^t$ sampled up to this point allows $I$ to be reached ({\em i.e.}, it is still possible that $\pi_i^t\in\Pi_i(I)$), then an action is selected according to the internal-regret minimizer $\rmi_I$ (Line~\ref{line:internal}). Otherwise, if $I$ cannot be reached through the (possibly partial) plan $\pi_i^t$, then we let $\sigma_I^t$ be the unique sequence in $\Sigma_i^c(I)$ whose actions are prescribed by $\pi_i^t$ (Line~\ref{line:sigmaT}). In this case, the player follows the strategy recommended by the external-regret minimizer $\rme_{\sigma_I^t,I}$ (Line~\ref{line:external}).
In the update procedure, the regret minimizers are fed with the vectors $\hat u_I^t$, which, with an abuse of notation, denote the vectors whose components are defined by the values of the corresponding parameterized utility functions $\hat u_I^t$ in Equation~\eqref{parametrized_util}.
In particular, for each $I\in\I_i$, the internal-regret minimizer $\rmi_I$ observes the utility vector $\hat u_I^t$ only if the sampled plan $\pi_i^t$ allows to reach infoset $I$, while each external-regret minimizer $\rme_{\sigma,I}$ is updated only if $\pi_i^t$ prescribes all the actions in the corresponding sequence $\sigma$ (Line~\ref{line:up_int}~and~Line~\ref{line:up_ext}, respectively).

The crucial insight is that for each infoset $I \in \I_i$, no matter the action selected at $I$, only one of the regret minimizers will receive a non-zero utility.
Consequently, only one of the regret minimizers can cumulate regret at time $t$, and that is the regret whose recommendation we follow.
Therefore, it is possible to show that the empirical frequency of play $\bar \mu^T$ obtained via \icfr~converges almost surely to an EFCE.
We start with the following auxiliary result.

\begin{restatable}{lemma}{goingUp}\label{lemma:going_up}
	For any $I,J\in \I_i: I\preceq J$, if $\hat R_{\sigma,J}^T=o(T)$ for all $\sigma=(I,a)\in\Sigma_i$ then $\hat R_{\sigma(I),J}^T=o(T)$.
\end{restatable}

Then, our main result reads as follows:

\begin{restatable}{theorem}{icfrTh}\label{th:icfr}
	When all the players play according to \emph{\icfr}, $\bar \mu^T$ converges almost surely to an EFCE.
\end{restatable}

%\subsection{Example}\label{subsec:example}
\paragraph{Example}

\begin{wrapfigure}{r}{6cm}
	\hfill
	\scalebox{.8}{
		{\begin{minipage}[b]{3cm}\centering%
				\def\done{.7*1.2}
				\def\dtwo{.40*1.2}
				\def\dleaf{.25*1.2}
				\def\dvert{.8*1.2}
				\begin{tikzpicture}[baseline=-1.1cm,scale=.95]
				\node[fill=black,draw=black,circle,inner sep=.5mm] (A) at (0, 0) {};
				\node[fill=black,draw=black,circle,inner sep=.5mm] (X) at ($(-\done,-\dvert)$) {};
				\node[fill=white,draw=black,inner sep=.5mm] (Y) at ($(\done,-\dvert)$) {};
				\node[fill=white,draw=black,inner sep=.5mm] (B) at ($(X) + (-\dtwo, -\dvert)$) {};
				\node[fill=white,draw=black,inner sep=.5mm] (C) at ($(X) + (\dtwo, -\dvert)$) {};
				
				\draw[semithick] (A) --node[fill=white,inner sep=.9] {\scriptsize$a$} (X);
				\draw[semithick,teal] (A) --node[fill=white,inner sep=.9] {\scriptsize$b$} (Y);
				\draw[semithick] (X) --node[fill=white,inner sep=.9] {\scriptsize$c$} (B);
				\draw[semithick] (X) --node[fill=white,inner sep=.9] {\scriptsize$d$} (C);
				
				\draw[purple!60!white] (X) circle (.2);
				\node[purple!90!white]  at ($(X) + (0,.38)$) {$J$};
				
				\draw[blue!60!white] (A) circle (.2);
				\node[blue!90!white]  at ($(A) + (0,.38)$) {$I$};
				\end{tikzpicture}
		\end{minipage}}
		\hspace{.2cm}
		{\begin{minipage}[b]{4cm}\centering
				\setlength{\tabcolsep}{1pt}
				\begin{tabular}{c@{\hskip 6pt}cccc}
					&\multicolumn{4}{c}{Trigger sequence} \\
					& $(I,a)$ & $(I,b)$ & $(J,c)$ & $(J,d)$ \\[-.2mm]
					\midrule\\[-3mm]
					$I$  & \textcolor{blue}{$\hat R_{a,I}^T$} & \textcolor{blue}{$\hat R_{b,I}^T$} & $\times$ & $\times$  \\[1mm]
					$J$  & $\hat R_{a,J}^T$ & \textcolor{teal}{$\hat R_{b,J}^T$} & \textcolor{purple}{$\hat R_{c,J}^T$} & \textcolor{purple}{$\hat R_{d,J}^T$}  \\[-.2mm]
					\bottomrule
				\end{tabular}
	\end{minipage}}}
	\caption{(Left) EFG with two infosets $I$ and $J$ of player $i$.
		(Right) The laminar subtree regrets.
	}
	\label{fig:example_icfr}
\end{wrapfigure}

We provide a simple example illustrating the key ideas of ICFR.
Figure~\ref{fig:example_icfr}--Left describes an EFG with two infosets $I,J$ of the same player (player $i$). 
Even in such a simple setting \icfr~has to ensure that six laminar subtree regrets are properly minimized (see Figure~\ref{fig:example_icfr}--Right).
To simplify the notation, throughout the example we write $\hat R_{a,I}^T$ in place of $\hat R_{(I,a),J}^T$ (the remaining regrets are treated analogously).
\icfr~instantiates one internal-regret minimizer for each infoset of player $i$. We denote them by $\rmi_I$ and $\rmi_J$, respectively. 
Then, we observe that $\Sigma_i^c(J)=\{(I,b)\}$, because $b$ is the only action of player $i$ satisfying the following conditions: (i) it departs from an infoset which is on the path from the root node to $J$ and (ii) if player $i$ selected $b$ at infoset $I$, she would no longer be able to reach $J$.
Therefore, \icfr~instantiates the external-regret minimizer $\rme_{b,J}$.

Suppose to be at iteration $t$ of~\icfr.
The sampling procedure starts from infoset $I$. Being the root of the EFG, $I$ is always reached by player $i$. 
Therefore, an action is selected following the recommendation of the internal-regret minimizer $\rmi_I$.
During the update procedure, $\rmi_I$ is provided with the utility resulting from the normal-form plan $\pi_i^t$ obtained from the sampling procedure.
Intuitively, this ensures that $\hat R_{a,I}^T$ and $\hat R_{b,I}^T$ are small.
Now, there are two possibilities: %either $\pi_i^t(I)=a$ (and $\pi_i^t \in \Pi_i(J)$) or $\pi_i^t(I)=b$ (and $\pi_i^t \notin \Pi_i(J)$) .

\textbf{Case $\pi_i^t(I)=a$.}
The partial plan $\pi_i^t$ allows $J$ to be reached. 
Therefore, at $J$, an action is chosen according to the strategy recommended by $\rmi_{J}$.
Then, in the update procedure, the internal-regret minimizer $\rmi_{J}$ is provided with the observed utility, while the external-regret minimizer is not updated.
This ensures that $\hat R_{c,J}^T$ and $\hat R_{d,J}^T$ are managed properly.
By Equation~\ref{internal_laminar}, the choice at $t$ does not impact $\hat R_{b,J}^T$ since $\pi_i^t\notin\Pi_i(I,b)$, while $\hat R_{a,J}^T$ is affected by the choice at $J$ because $a\in\sigma(J)$.
The internal-regret minimizer $\rmi_{J}$ guarantees that $\hat R_{c,J}^T=o(T)$ and $\hat R_{d,J}^T=o(T)$. 
Then, by using Lemma~\ref{lemma:going_up}, we have that $\hat R_{a,J}^T=o(T)$ holds as well.

\textbf{Case $\pi_i^t(I)=b$.}
We have that $\sigma_J^t=(I,b)$. An action at $J$ is sampled according to the external-regret minimizer $\rme_{b,J}$, which is then provided with the observed utility (the internal-regret minimizer $\rmi_{J}$ is not updated).
This ensures that the increase in $\hat R_{b,J}^T$ is small.
The other regret terms are not impacted by the choice at $t$.

    \section{Experimental evaluation}

We evaluate the convergence of~\icfr~on the standard benchmark games for the computation of correlated equilibria. We use parametric instances from four different multi-player games:  Kuhn poker~\citep{kuhn1950simplified}, Leduc poker~\cite{southey2005bayes}, Goofspiel~\citep{ross1971goofspiel}, and Battleship~\citep{farina2019correlation}.
Instances of the Kuhn, Leduc, and Goofspiel games are parametric in the number of players $p$ and in the number of card ranks $r$.
To increase the readability, we denote by \textsc{K}$p$.$r$ the Kuhn poker instance with $p$ players and $r$ ranks (the other instances are treated analogously). 
Our Battleship instance (denoted by \textsc{Bs}) has a grid of size $2\times 2$ and maximum number of rounds per player equal to $3$.
%
%
% All the instances we employed can be found in~\cref{fig:dimensions} in~\cref{subsec:exp_games}. 
%
A detailed description of the games is provided in~\cref{subsec:exp_games}.
We use {\em Regret matching}~\cite{hart2000simple} for external-regret minimizers, and the no-internal-regret algorithm by~\citet{blum2007external} for internal-regret minimizers. All experiments are run on a 64-core machine with 512 GB of RAM.

\paragraph{Convergence of~\icfr}
\cref{fig:exp}--Center displays the maximum deviation $\delta(\bar\mu^T)$ as a function of the number of rounds $T$.
According to~\cref{eq:efce_simp_eps}, the strategy $\bar\mu^T$ is guaranteed to be a $\delta(\bar\mu^T)$-EFCE. 
We set a maximum number of $10^4$ iterations and, for each instance, we provide the average and the standard deviation computed over $50$ different seeds.
First, we notice that \icfr~attains roughly an empirical convergence rate of $O(1/T)$. The performance over the Battleship instance suggests that equilibria with large support size are significantly more challenging to be computed.
Second, we remark that, unlike recent algorithms for computing EFCEs by~\citet{farina2019correlation,farina2019efficient},~\icfr~can be applied to games with more than two players including chance.
Moreover, since $\textnormal{EFCE} \subseteq\textnormal{EFCCE}\subseteq\textnormal{NFCCE}$,~\icfr~also provides a flexible way to compute $\epsilon$-EFCCEs and $\epsilon$-NFCCEs. 
In the former case, the only known algorithm can only handle games with two players and no chance~\cite{farina2019coarse}. In the latter case, the recent algorithms by~\citet{celli2019learning} are significantly outperformed.
For example, previous algorithms cannot reach a $0.1$-NFCCE in less than $24$h on a Leduc instance with $1200$ total infosets and a one-bet maximum per bidding round.~\icfr~reaches $\epsilon=0.1$ in around $9$h on an arguably more complex Leduc instance ({\em i.e.}, more than $9$k total infosets and a two-bet maximum per round).
Further details on the computation of EFCCEs and NFCCEs are provided in~\cref{subsec:more_exp}, together with the plots of the decoupled EFCE deviations of each player.

\paragraph{Social Welfare}
\cref{fig:exp}--Right provides a visual depiction of the {\em quality} of the solutions attained by \icfr~in terms of their social welfare.
The figure displays the payoffs obtained for $100$ different seeds in a two-player Goofspiel instance without chance ({\em i.e.,} the prize deck is sorted).

\begin{figure}
	\centering
	\hspace{-12mm}
	{\begin{minipage}{4cm}\centering
			\small
			\setlength{\tabcolsep}{1pt}
			\vspace{-5mm}
			\begin{tabular}{c@{\hskip 6pt}|c|cc}
				& \textbf{Pl.} & \textbf{Info.} & \textbf{Seq.} \\[-.5mm]
				\midrule\\[-4mm]
				$\textsc{K}3.3$ & $3$ & $12$ & $25$   \\[-.2mm]
				$\textsc{K}3.4$ & $3$ & $16$ & $33$   \\[-.2mm]
				$\textsc{G}2.3$ & $2$ & $213$ & $262$   \\[-.2mm]
				$\textsc{G}2.4$ & $2$ & $8716$ & $10649$   \\[-.2mm]
				$\textsc{G}3.3$ & $3$ & $837$ & $934$   \\[-.2mm]
				$\textsc{L}3.3$ & $3$ & $3294$ & $7687$   \\[-.2mm]
				\midrule
				\multirow{2}{*}{\textsc{Bs}} & \multirow{2}{*}{$2$} & 1413 & 2965 \\
				&& 1873 & 4101 \\
				\bottomrule
			\end{tabular}
	\end{minipage}}
	\hspace{-4mm}
	{\begin{minipage}{4cm}
			\vspace{-.1mm}
			\begin{tikzpicture}
			\node[anchor=south west,inner sep=0] (image) at (0,0){\includegraphics[width=1.5\textwidth]{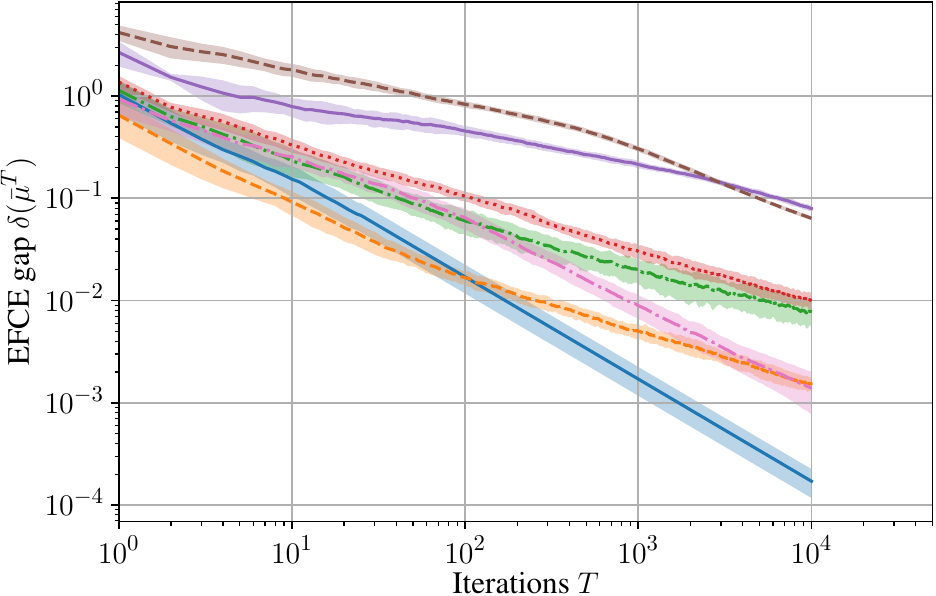}};
			
			\begin{scope}[
			x={(image.south east)},
			y={(image.north west)}
			]
			\node [black, font=\scriptsize] at (.91,0.66) {\textsc{Bs}};
			\node [black, font=\scriptsize] at (.92,0.6) {$\textsc{L}3.3$};
			\node [black, font=\scriptsize] at (.92,0.52) {$\textsc{G}2.4$};
			\node [black, font=\scriptsize] at (.92,0.46) {$\textsc{G}2.3$};
			\node [black, font=\scriptsize] at (.92,0.30) {$\textsc{G}3.3$};
			\node [black, font=\scriptsize] at (.92,0.35) {$\textsc{K}3.4$};
			\node [black, font=\scriptsize] at (.92,0.18) {$\textsc{K}3.3$};
			\end{scope}
			
			\end{tikzpicture}
	\end{minipage}}
	\hspace{20mm}
	\begin{minipage}{4cm}
		\vspace{-.1mm}
		\includegraphics[width=1.15\textwidth]{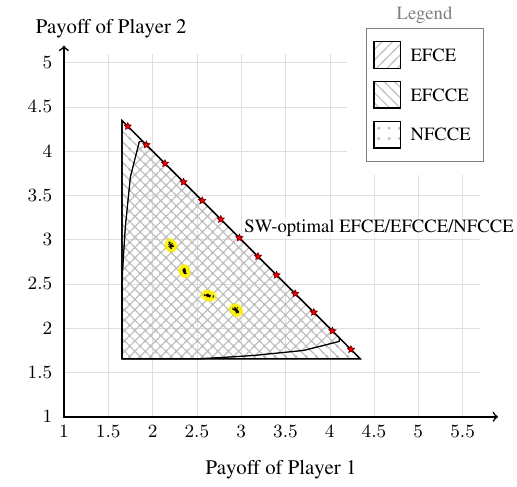}
	\end{minipage}
	\caption{(Left) Dimension of the game instances in terms of number of players and infosets/sequences for each player. (Center) Convergence of~\icfr. (Right) Social welfare attained at different $\epsilon$-EFCEs computed via~\icfr~(black dots corresponds to different seeds).}
	\label{fig:exp} 
\end{figure}

    \section*{Broader Impact}

Correlated equilibria provide an appropriate solution concept for coordination problems in which agents have arbitrary utilities, and may work towards different objectives. 
The study of uncoupled dynamics converging to correlated equilibria in problems with sequential actions and hidden information lays new theoretical foundations for multi-agent reinforcement learning problems. 
Most of the work in the multi-agent reinforcement learning community either studies fully competitive settings, where agents play selfishly to reach a Nash equilibrium, or fully cooperative scenarios in which agents have the exact same goals. Our work could enable techniques that are in-between these two extremes: agents have arbitrary objectives, but coordinate their actions towards an equilibrium with some desired properties. 

As we argued in the paper, the social welfare that can be attained via a Nash equilibrium (that is, by playing selfishly) may be significantly lower than what can be achieved via a correlated equilibrium. We provided some empirical evidences that \icfr~computes equilibria which attain a social welfare `not too far' from the optimal one. This could have an arguably positive societal impact when applied to real economic problems. 
However, further research in this direction is required to prevent `winner-takes-all' scenarios in problems with an unbalanced reward structure where equilibria with high social welfare may just award players with the largest utilities at the expense of the others. This could provide a way to reach {\em fair} equilibria both in theory and in practice.
    
    \begin{ack}
    	
    	This work is based on work supported by the Italian MIUR PRIN 2017 Project ALGADIMAR ``Algorithms, Games, and Digital Market'', the National Science Foundation under grants IIS-1718457, IIS-1617590, IIS-1901403, and CCF-1733556, and the ARO under awards W911NF-17-1-0082 and W911NF2010081. Gabriele Farina is supported by a Facebook fellowship.
    \end{ack}

    \bibliographystyle{plainnat}
    \bibliography{refs}

	\clearpage
    \appendix
    \addtolength{\hoffset}{-1cm}
    \addtolength{\textwidth}{2cm}
    \addtolength{\hsize}{2cm}
    \addtolength{\linewidth}{2cm}
    \addtolength{\columnwidth}{2cm}
    \appendix
\section{Extensive-form correlated equilibrium}\label{appendix:efce}

In the context of EFGs, the two most widely adopted notions of correlated equilibrium are the {\em normal-form correlated equilibrium} (NFCE)~\citep{aumann1974subjectivity} and the {\em extensive-form correlated equilibrium} (EFCE)~\citep{von2008extensive}.
In the former, the mediator draws and recommends a complete normal-form plan to each player before the game starts.
Then, each player decides whether to follow the recommended plan or deviate to an arbitrary strategy she desires.
In an EFCE, the mediator draws a normal-form plan for each player before the beginning of the game, but she does not immediately reveal it to each player.
Instead, the mediator incrementally reveals individual moves as players reach new infosets.
At any infoset, the acting player is free to deviate from the recommended action, but doing so comes at the cost of future recommendations, which are no longer issued if the player deviates.

In an EFCE, players know less about the normal-form plans that were sampled by the mediator than in an NFCE, where the whole normal-form plan is immediately revealed.
Therefore, by exploiting an EFCE, the mediator can more easily incentivize players to follow strategies that may hurt them, as long as players are indifferent as to whether or not to follow the recommendations.
This is beneficial when the mediator wants to maximize, {\em e.g.}, the social-welfare of the game.

A {\em coarse} correlated equilibrium enforces protection against deviations which are independent of the recommended move.
{\em Normal-form coarse correlated equilibria} (NFCCEs)~\citep{moulin1978,celli2018computing} and {\em extensive-form coarse correlated equilibria} (EFCCEs)~\citep{farina2019coarse} are the coarse equivalent of NFCE and EFCE, respectively.
For arbitrary EFGs with perfect recall, the following inclusion of the set of equilibria holds:
$\textnormal{NFCE}\subseteq \textnormal{EFCE} \subseteq\textnormal{EFCCE}\subseteq\textnormal{NFCCE}$~\citep{von2008extensive,farina2019coarse}.

Appendix~\ref{subsec:def_efce} provides a suitable formal definition of the set of EFCEs via the notion of {\em trigger agent} (originally introduced by~\citet{gordon2008no}~and~\citet{dudik2009sampling}).
Finally, Appendix~\ref{subsec:computing_correlated_eq} summarizes existing approaches for computing EFCEs.

\subsection{Formal definition of the set of EFCEs}\label{subsec:def_efce}
The definition requires the following notion of trigger agent, which, intuitively, is associated to each player and each of her sequences of action recommendations.

\begin{definition}[Trigger agent for EFCE]
	Given a player $i \in \pl$, a sequence $\sigma = (I,a) \in \Sigma_i$
	%
	%, an infoset $I \in \I_i$, an action $a \in A(I)$
	%
	, and a probability distribution $\hat \mu_i \in \Delta_{\Pi_i(I)}$, an \emph{$(\sigma, \hat \mu_i)$-trigger agent for player $i$} is an agent that takes on the role of player $i$ and commits to following all recommendations unless she reaches $I$ and gets recommended to play $a$.
	If this happens, the player stops committing to the recommendations and plays according to a plan sampled from $\hat \mu_i$ until the game ends.
\end{definition}

It follows that joint probability distribution $\mu \in \Delta_\Pi$ is an EFCE if, for every $i \in \pl$, player $i$'s expected utility when following the recommendations is at least as large as the expected utility that any $(\sigma, \hat \mu_i)$-trigger agent for player $i$ can achieve (assuming the opponents' do not deviate).

Given $\sigma = (I,a) \in \Sigma_i$, in order to express the expected utility of a $(\sigma, \hat \mu_i)$-trigger agent, it is convenient to define the probability of the game ending in each terminal node $z \in Z$.
Three cases are possible.
In the first one, $z \in Z(I,a)$. 
The probability of reaching $z$ given the joint probability distribution $\mu\in\Delta_\Pi$ and a $(\sigma, \hat \mu_i)$-trigger agent is defined as:
\begin{equation}
p_{\mu,\hat\mu_i}^\sigma(z) \coloneqq \left( \sum_{\substack{\pi_i \in \Pi_i(\sigma)\\\pi_{-i} \in \Pi_{-i}(z)}} \mu (\pi_i, \pi_{-i}) \right) \left( \sum_{\hat \pi_i \in \Pi_i(z)} \hat \mu_i(\hat \pi_i)  \right) p_c(z),
\end{equation}
which accounts for the fact that the agent follows recommendations until she receives the recommendation of playing $a$ at $I$, and, thus, she `gets triggered' and plays according to $\hat \pi_i$ sampled from $\hat \mu_i$ from $I$ onwards.
The second case is $z \in Z^c(I,a)$, which is reached with probability:
\begin{equation}
y_{\mu,\hat\mu_i}^\sigma(z) \coloneqq \left( \sum_{\substack{\pi_i \in \Pi_i(\sigma)\\\pi_{-i} \in \Pi_{-i}(z)}} \mu (\pi_i, \pi_{-i}) \right) \left( \sum_{\hat \pi_i \in \Pi_i(z)} \hat \mu_i(\hat \pi_i)  \right) p_c(z) + \left( \sum_{\pi \in \Pi(z)} \mu(\pi) \right) p_c(z),
\end{equation}
where the first term accounts for the event that $z$ is reached when the agent `gets triggered', while the second term is the probability of reaching $z$ while not being triggered (notice that the two events are independent).
Finally, the third case is when $z \in Z \setminus Z(I)$ and the infoset $I$ is never reached.
Then, the probability of reaching $z$ is defined as:
\begin{equation}\label{eq:q}
q_\mu(z) \coloneqq \left( \sum_{\pi \in \Pi(z)} \mu(\pi) \right) p_c(z).
\end{equation}

By exploiting the above definitions, the definition of EFCE reads as follows.

\begin{definition}[Extensive-form correlated equilibrium]\label{def:efce}
	An EFCE of an EFG is a probability distribution $\mu \in \Delta_\Pi$ such that, for every $i \in  \pl$ and $(\sigma, \hat \mu_i)$-trigger agent for player $i$, with $\sigma =(I,a) \in \Sigma_i$, it holds:
	\begin{equation}\label{eq:efce}
	\sum_{z \in Z} \left( \sum_{\pi \in \Pi(z)} \mu(\pi)  \right) p_c(z) u_i(z) \geq \!\!\! \sum_{z \in Z(I,a)} \!\!\! p_{\mu,\hat\mu_i}^\sigma(z) u_i(z) + \!\!\! \sum_{z \in Z^c(I,a)} \!\!\! y_{\mu,\hat\mu_i}^\sigma(z) u_i(z) +\!\!\! \sum_{z \in Z \setminus Z(I)} \!\!\! q_{\mu}(z) u_i(z).
	\end{equation}
\end{definition}

Noticing that the left-hand side of Equation~\eqref{eq:efce} is equal to $\sum_{z \in Z} q_\mu(z) u_i(z)$ and that $y_{\mu,\hat\mu_i}^\sigma(z) = p_{\mu,\hat\mu_i}^\sigma(z) + q_\mu(z)$, we can rewrite Equation~\eqref{eq:efce} as follows:
\begin{equation}
\sum_{z \in Z(I,a)} q_\mu(z) u_i(z) \geq \sum_{z \in Z(I)} p_{\mu,\hat\mu_i}^\sigma(z) u_i(z).
\end{equation}

A probability distribution $\mu\in\Delta_\Pi$ is said to be an $\epsilon$-EFCE if, for every $i \in  \pl$ and $(\sigma, \hat \mu_i)$-trigger agent for player $i$, with $\sigma =(I,a) \in \Sigma_i$, it holds:
\begin{equation}
\sum_{z \in Z(I,a)} q_\mu(z) u_i(z) \geq \sum_{z \in Z(I)} p_{\mu,\hat\mu_i}^\sigma(z) u_i(z) - \epsilon.
\end{equation}

\subsection{Computation of EFCEs}\label{subsec:computing_correlated_eq}

The problem of computing an optimal EFCE in extensive-form games with more than two players and/or chance moves is known to be \textsf{NP}-hard~\citep{von2008extensive}.
However,~\citet{huang2008computing} show that the problem of finding {\em one} EFCE can be solved in polynomial time via a variation of the {\em Ellipsoid Against Hope} algorithm~\citep{papadimitriou2008,jiang2015polynomial}.
This holds for arbitrary EFGs with multiple players and/or chance moves.
Unfortunately, that algorithm is mainly a theoretical tool, and it is known to have limited scalability beyond toy problems.
\citet{dudik2009sampling} provide an alternative sampling-based algorithm to compute EFCEs.
However, their algorithm is {\em centralized} and based on MCMC sampling which may limit its practical appeal.
Our framework is arguably simpler and based on the classical {\em counterfactual regret minimization} algorithm~\citep{zinkevich2008regret,farina2018online}.
Moreover, our framework is fully {\em decentralized} since each player, at every decision point, plays so as to minimize her internal/external regret.

If we restrict our attention to two-player perfect-recall games without chance moves, than the problem of determining an optimal EFCE can be characterized through a succint linear program with polynomial size in the game description~\citep{von2008extensive}.
In this setting,~\citet{farina2019correlation} show that the problem of computing an EFCE can be formulated as the solution to a bilinear saddle-point problem, which they solve via a subgradient descent method.
Moreover,~\citet{farina2019efficient} design a regret minimization algorithm suitable for this specific scenario. 
In a recent paper, \citet{Farina20:Polynomial} showed that that an optimal EFCE,
EFCCE and NFCCE can be computed in polynomial time in the game size in
two-player general-sum games that satisfy a condition known as \emph{triangle-freeness}. The triangle-freeness condition holds, for example, when all chance moves are \emph{public}, that is, both players observe all chance moves.

\section{Omitted proofs}\label{appendix:proofs}

\subsection{Proofs for Section~\ref{sec:efce_regret}}

The following auxiliary result is exploited in the proof of Theorem~\ref{th:efce}.

\begin{lemma}\label{lemma:rho}
	For every iteration $t = 1, \ldots, T$, player $i\in\pl$, plan $\hat\pi_i\in\Pi_i$, joint plan $\pi^t=(\pi^t_i,\pi^t_{-i})\in\Pi$, and infoset $I\in\I_i$, the following holds:
	\[
	V_I^t(\hat\pi_i) - V_I^t(\pi^t_i)= \sum_{z\in Z(I)}\left( \rho_{I\to z}^{(\hat\pi_{i },\pi^t_{-i})} - \rho^{\pi^t}_{I\to z}\right) p_c(z) u_i(z).
	\]
\end{lemma}
\begin{proof}
	Given an arbitrary infoset $I\in\I_i$, the set of terminal nodes immediately reachable from $I$ through action $a\in A(I)$ is defined as
	\[
	Z^{\textsc{I}}(I,a)\defeq Z(I,a) \setminus \bigcup_{J \in \mathcal{C}(I,a)} Z(J).
	\]
	
	By expanding $V_I^t(\hat \pi_i)$ according to its definition (Equation~\eqref{eq:v_definition}) and by substituting the definition of immediate utility vector ${u}_i^t$ we obtain that
	\begin{align*}
	V_I^t(\hat \pi_{i}) &= u^t_{i}[I,\hat \pi_{i}(I)] + \sum_{J \in\mathcal{C}({I,\hat \pi_{i}(I)})} V_{J}^t(\hat \pi_{i })\\
	& = \sum_{z \in Z^{\textsc{I}}(I,\hat\pi_i(I)) } \one[\pi_{-i}^t\in\Pi_{-i}(z)]\, p_c(z) u_i(z) + \sum_{J \in\mathcal{C}({I,\hat \pi_{i}(I)})} V_{J}^t(\hat \pi_{i })\\
	& = \sum_{z\in Z(I)} \rho_{I\to z}^{\hat \pi_i}\one[\pi_{-i}^t\in\Pi_{-i}(z)]\, p_c(z) u_i(z),
	\end{align*}
	where the last expression is obtained by expanding recursively the terms $V_{J}^t(\hat \pi_{i })$.
	By definition, $\rho_{I\to z}^{(\hat\pi_{i },\pi_{-i}^t)} = \rho_{I\to z}^{\hat\pi_{i}}\cdot \one[\pi_{-i}^t\in\Pi_{-i}(z)]$.
	Therefore, we can write $V_I^t(\hat \pi_{i})=\sum_{z\in Z(I)}\rho_{I\to z}^{(\hat\pi_{i },\pi_{-i}^t)}p_c(z)u_i(z)$.
	Analogously, by expanding $V_I^t(\pi_{i}^t)$, we obtain that $V_I^t( \pi_{i}^t)=\sum_{z\in Z(I)}\rho_{I\to z}^{\pi^t}p_c(z)u_i(z)$.
	This concludes the proof.
\end{proof}

\efceTh*
\begin{proof}\allowdisplaybreaks
    Fix any player $i \in \pl$ and any sequence $\sigma = (I,a)\in\Sigma_i$ for her. From Lemma~\ref{lemma:rho}, the regret $R^T_\sigma$ is
	\begin{align*}
	R^T_\sigma&=\max_{\hat\pi_i\in\Pi_i(I)}\sum_{t=1}^T\one[\pi^t_i\in\Pi_i(\sigma)]\mleft(\sum_{z\in Z(I)}\left( \rho_{I\to z}^{(\hat\pi_{i },\pi^t_{-i})} - \rho^{\pi^t}_{I\to z}\right) p_c(z) u_i(z)\mright)
	\\&=
	\max_{\hat\pi_i\in\Pi_i(I)}\sum_{t=1}^T \sum_{\pi\in\Pi} \one[\pi=\pi^t] \mleft(\one[\pi_i\in\Pi_i(\sigma)] \mleft(\sum_{z\in Z(I)}\left( \rho_{I\to z}^{(\hat\pi_{i },\pi_{-i})} - \rho^{\pi}_{I\to z}\right) p_c(z) u_i(z)\mright)\mright)
	\\&=
	\max_{\hat\pi_i\in\Pi_i(I)}\sum_{\pi\in\Pi} \one[\pi_i\in\Pi_i(\sigma)] \mleft(\mleft(\sum_{t=1}^T \one[\pi=\pi^t] \mright)  \mleft(\sum_{z\in Z(I)}\left( \rho_{I\to z}^{(\hat\pi_{i },\pi_{-i})} - \rho^{\pi}_{I\to z}\right) p_c(z) u_i(z)\mright) \mright).
	\end{align*}
	By using the definition of empirical frequency of play, we can write $\sum_{t=1}^T \one[\pi=\pi^t] = T \bar{\mu}^T(\pi)$. Hence,
	\begin{align*}
	R^T_\sigma&=
	T \max_{\hat\pi_i\in\Pi_i(I)} \sum_{\pi\in\Pi} \one[\pi_i\in\Pi_i(\sigma)] \mleft(\bar\mu^T(\pi) \mleft(\sum_{z\in Z(I)}\left( \rho_{I\to z}^{(\hat\pi_{i },\pi_{-i})} - \rho^{\pi}_{I\to z}\right) p_c(z) u_i(z)\mright) \mright)
    \\&=
    T \max_{\hat\pi_i\in\Pi_i(I)}\sum_{\substack{\pi_i\in\Pi_i(\sigma)\\\pi_{-i}\in\Pi_{-i}}} \bar\mu^T(\pi) \mleft(\sum_{z\in Z(I)}\left( \rho_{I\to z}^{(\hat\pi_{i },\pi_{-i})} - \rho^{\pi}_{I\to z}\right) p_c(z) u_i(z)\mright)
    \\&=
    T \max_{\hat\pi_i\in\Pi_i(I)}\sum_{z\in Z(I)} \sum_{\substack{\pi_i\in\Pi_i(\sigma)\\\pi_{-i}\in\Pi_{-i}}} \bar\mu^T(\pi) \left( \rho_{I\to z}^{(\hat\pi_{i },\pi_{-i})} - \rho^{\pi}_{I\to z}\right) p_c(z) u_i(z).
    \end{align*}
    Using the definition of the $\rho_{I\to z}$ symbols, that is,
	\[
		\rho_{I\to z}^{(\hat\pi_{i },\pi_{-i})} = \rho_{I\to z}^{\hat\pi_{i}}\cdot \one[\pi_{-i}\in\Pi_{-i}(z)], \qquad
        \rho_{I\to z}^{\pi} = \rho_{I\to z}^{\pi_{i}}\cdot \one[\pi_{-i}\in\Pi_{-i}(z)],
	\]
    we further obtain
    \begin{align*}
	R^T_\sigma &=
    T \max_{\hat\pi_i\in\Pi_i(I)}\sum_{z\in Z(I)} \sum_{\substack{\pi_i\in\Pi_i(\sigma)\\\pi_{-i}\in\Pi_{-i}}} \bar\mu^T(\pi) \left( \rho_{I\to z}^{\hat\pi_{i}} - \rho^{\pi_i}_{I\to z}\right)\one[\pi_{-i}\in\Pi_{-i}(z)] p_c(z) u_i(z)
    \\&=
	T \max_{\hat\pi_i\in\Pi_i(I)}\sum_{z\in Z(I)} \sum_{\substack{\pi_i\in\Pi_i(\sigma)\\\pi_{-i}\in\Pi_{-i}(z)}} \bar\mu^T(\pi) \left( \rho_{I\to z}^{\hat\pi_{i}} - \rho^{\pi_i}_{I\to z}\right)p_c(z) u_i(z)
    \\&=
	T \vast(\underbrace{\max_{\hat\pi_i\in\Pi_i(I)}\sum_{z\in Z(I)} \sum_{\substack{\pi_i\in\Pi_i(\sigma)\\\pi_{-i}\in\Pi_{-i}(z)}} \bar\mu^T(\pi) \rho_{I\to z}^{\hat\pi_{i}}\,p_c(z) u_i(z)}_{\circled{B}}\vast) - T \vast(\underbrace{\sum_{z\in Z(I)} \sum_{\substack{\pi_i\in\Pi_i(\sigma)\\\pi_{-i}\in\Pi_{-i}(z)}} \bar\mu^T(\pi) \rho^{\pi_i}_{I\to z}\,p_c(z) u_i(z)}_{\circled{C}}\vast).
	\end{align*}
	
	We now analyze \circled{B} and \circled{C} separately.
	\begin{itemize}[nolistsep,itemsep=1mm,leftmargin=8mm]
		\item[\circled{B}] By convexity, we have:
            \begin{align*}
	           \circled{B} &= \max_{\hat\mu_i \in \Delta_{\Pi_i(I)}} \mleft\{\sum_{\hat\pi_i \in \Pi_i(I)} \hat\mu_i(\hat\pi_i)\mleft( \sum_{z\in Z(I)} \sum_{\substack{\pi_i\in\Pi_i(\sigma)\\\pi_{-i}\in\Pi_{-i}(z)}} \bar\mu^T(\pi) \rho_{I\to z}^{\hat\pi_{i}}\,p_c(z) u_i(z)\mright)\mright\}\\
                    &= \max_{\hat\mu_i \in \Delta_{\Pi_i(I)}} \mleft\{ \sum_{z\in Z(I)} \mleft(\sum_{\hat\pi_i \in \Pi_i(I)} \hat\mu_i(\hat\pi_i) \rho_{I\to z}^{\hat\pi_{i}} \mright)\mleft( \sum_{\substack{\pi_i\in\Pi_i(\sigma)\\\pi_{-i}\in\Pi_{-i}(z)}} \bar\mu^T(\pi) p_c(z) u_i(z)\mright)\mright\}.
	       \end{align*}
            Since $\hat\pi_i \in \Pi_i(I)$ and $z \in Z(I)$, $\rho_{I\to z}^{\hat\pi_i} = \one[\hat\pi_i \in \Pi_i(z)]$. So,
            \begin{align*}
	           \circled{B} &= \max_{\hat\mu_i \in \Delta_{\Pi_i(I)}} \mleft\{ \sum_{z\in Z(I)} \mleft(\sum_{\hat\pi_i \in \Pi_i(I)} \hat\mu_i(\hat\pi_i)\,\one[\hat\pi_i \in \Pi_i(z)] \mright)\mleft( \sum_{\substack{\pi_i\in\Pi_i(\sigma)\\\pi_{-i}\in\Pi_{-i}(z)}} \bar\mu^T(\pi) p_c(z) u_i(z)\mright)\mright\}\\
                    &= \max_{\hat\mu_i \in \Delta_{\Pi_i(I)}} \mleft\{ \sum_{z\in Z(I)} \mleft(\sum_{\hat\pi_i \in \Pi_i(z)} \hat\mu_i(\hat\pi_i)\mright)\mleft( \sum_{\substack{\pi_i\in\Pi_i(\sigma)\\\pi_{-i}\in\Pi_{-i}(z)}} \bar\mu^T(\pi) \mright) p_c(z) u_i(z)\mright\}\\
                    &= {\max_{\hat\mu_i \in \Delta_{\Pi_i(I)}} \mleft\{ \sum_{z\in Z(I)} p^{\sigma}_{\bar\mu^T\!\!,\,\hat\mu_i}(z) \, u_i(z)\mright\}}.\numberthis{eq:part B}
            \end{align*}

		\item[\circled{C}] Since $\pi_i \in \Pi_i(\sigma) \subseteq \Pi_i(I)$ and $z \in Z(I)$, $\rho_{I\to z}^{\pi_i} = \one[z \in Z(\sigma)]\cdot\one[\pi_i \in \Pi_i(z)]$. Therefore,
		\begin{align*}
		\circled{C} &= \sum_{z\in Z(I)} \sum_{\substack{\pi_i\in\Pi_i(\sigma)\\\pi_{-i}\in\Pi_{-i}(z)}} \bar\mu^T(\pi) \, \one[z \in Z(\sigma)]\one[\pi_i \in \Pi_i(z)]  \,   p_c(z) u_i(z)\\
		&= \sum_{z\in Z(I)} \mleft(\one[z \in Z(\sigma)] \sum_{\substack{\pi_i\in\Pi_i(\sigma)\\\pi_{-i}\in\Pi_{-i}(z)}} \bar\mu^T(\pi) \, \one[\pi_i \in \Pi_i(z)]  \,   p_c(z) u_i(z)\mright)\\
        &= \sum_{z\in Z(\sigma)} \mleft(\sum_{\substack{\pi_i\in\Pi_i(z)\\\pi_{-i}\in\Pi_{-i}(z)}} \bar\mu^T(\pi) \mright) p_c(z) u_i(z)\\
        &= {\sum_{z\in Z(\sigma)} q_{\bar\mu^T}(z) \, u_i(z)}.\numberthis{eq:part C}
		\end{align*}
	\end{itemize}
	
	Substituting the expressions in \eqref{eq:part B} and \eqref{eq:part C} into the expression for $R^T_\sigma$, we obtain
    \begin{equation}\label{eq:regret_pq}
        \frac{R^T_\sigma}{T} = \max_{\hat\mu_i \in \Delta_{\Pi_i(I)}} \mleft\{ \sum_{z\in Z(I)} p^{\sigma}_{\bar\mu^T\!\!,\,\hat\mu_i}(z) \, u_i(z)\mright\} - \sum_{z\in Z(\sigma)} q_{\bar\mu^T}(z) \, u_i(z).
    \end{equation}
    Finally, using the hypothesis, we can write
    \begin{align*}
        \epsilon &= \max_{i\in\pl}\max_{\sigma\in\Sigma_i} \frac{R^T_\sigma}{T} \\
            &= \max_{i\in\pl}\max_{\sigma\in\Sigma_i} \mleft\{ \max_{\hat\mu_i \in \Delta_{\Pi_i(I)}} \mleft\{ \sum_{z\in Z(I)} p^{\sigma}_{\bar\mu^T\!\!,\,\hat\mu_i}(z) \, u_i(z)\mright\} - \sum_{z\in Z(\sigma)} q_{\bar\mu^T}(z) \, u_i(z)\mright\}\\
            &= \delta(\bar\mu^T).
    \end{align*}
    This concludes the proof.
\end{proof}

\corefceconvergence*
\begin{proof}
	By~\cref{eq:regret_pq} we obtain:
    \begin{align*}
     0 &\ge \limsup_{T \to\infty} \frac{R^T_\sigma}{T} \\
       &= \limsup_{T \to\infty} \mleft(\max_{\hat\mu_i \in \Delta_{\Pi_i(I)}} \mleft\{ \sum_{z\in Z(I)} p^{\sigma}_{\bar\mu^T\!\!,\,\hat\mu_i}(z) \, u_i(z)\mright\} - \sum_{z\in Z(\sigma)} q_{\bar\mu^T}(z) \, u_i(z)\mright)\\
       &\ge \max_{\hat\mu_i \in \Delta_{\Pi_i(I)}} \mleft\{ \limsup_{T \to\infty}\sum_{z\in Z(I)} p^{\sigma}_{\bar\mu^T\!\!,\,\hat\mu_i}(z) \, u_i(z)\mright\} - \limsup_{T \to\infty} \sum_{z\in Z(\sigma)} q_{\bar\mu^T}(z)\, u_i(z),
    \end{align*}
    where the last inequality follows from swapping the order of $\limsup$ and $\max$. By definition of $\limsup$, for any $\epsilon > 0$, eventually $\delta(\bar{\mu}^T) < 0$ (more precisely: for any $\epsilon > 0$, there must be a $\tau = \tau(\epsilon)$ such that $\delta(\bar{\mu}^T) < \epsilon$ for all $T \ge \tau$), which means that eventually the empirical frequency of play $\bar{\mu}^T$ becomes an $\epsilon$-EFCE.
\end{proof}

\subsection{Proof for Section~\ref{sec:laminar}}

\subtreeDecomposed*
\begin{proof}\allowdisplaybreaks
	By using the recursive definitions of $R^T_{\sigma,I}$ and $V_I^t(\hat \pi_i)$, we get:
	\begin{align*}
	R_{\sigma,I}^T & = \max_{\hat \pi_i\in\Pi_i(I)}\mleft\{ \sum_{t=1}^T \one[\pi_i^t\in\Pi_i(\sigma)] \mleft( V_I^t(\hat \pi_i) - V^t_I(\pi^t_i) \mright)  \mright\}\\
	& = \max_{\hat \pi_i\in\Pi_i(I)}\mleft\{ \sum_{t=1}^T \one[\pi_i^t\in\Pi_i(\sigma)] V_I^t(\hat \pi_i) \mright\} - \sum_{t=1}^T \one[\pi_i^t\in\Pi_i(\sigma)] V^t_I(\pi^t_i)\\
	& = \max_{\hat \pi_i\in\Pi_i(I)}\mleft\{ \sum_{t=1}^T \one[\pi_i^t\in\Pi_i(\sigma)] \mleft( u^t_{i}[I,\hat \pi_{i}(I)] + \sum_{I' \in\mathcal{C}({I,\hat \pi_{i}(I)})} V_{I'}^t( \hat \pi_{i })\mright) \mright\}\\ 
	&\hspace{8.5cm}- \sum_{t=1}^T \one[\pi_i^t\in\Pi_i(\sigma)] V^t_I(\pi^t_i)\\
	& = \max_{a\in A(I)}\mleft\{ \sum_{t=1}^T\one[\pi_i^t\in\Pi_i(\sigma)] u_i^t[I,a] + \sum_{I'\in\mathcal{C}(I,a)} \max_{\hat\pi_i\in\Pi_i(I')}\mleft\{\sum_{t=1}^T\one[\pi_i^t\in\Pi_i(\sigma)]V_{I'}^t(\hat \pi_i)  \mright\}\mright\}\\
	&\hspace{8.5cm} - \sum_{t=1}^T \one[\pi_i^t\in\Pi_i(\sigma)] V^t_I(\pi^t_i)\\
	& = \max_{a\in A(I)}\mleft\{ \sum_{t=1}^T\one[\pi_i^t\in\Pi_i(\sigma)] u_i^t[I,a] + \sum_{I'\in\mathcal{C}(I,a)}\mleft( R_{\sigma,I'}^T +  \sum_{t=1}^T \one[\pi_i^t\in\Pi_i(\sigma)] V^t_{I'}(\pi_i^t) \mright)\mright\}\\
	&\hspace{8.5cm} - \sum_{t=1}^T \one[\pi_i^t\in\Pi_i(\sigma)] V^t_I(\pi^t_i),
	\end{align*}
	where the last step is by definition of subtree regret.
	By rewriting the above expression according to Equation~\eqref{parametrized_util} we get the result.
\end{proof}

\laminarDecomposition*
\begin{proof}\allowdisplaybreaks
	Consider an arbitrary sequence $\sigma=(J,a') \in \Sigma_i$ and infoset $I \in \mathcal{C}^\star(J)$.
	By Lemma~\ref{lemma:cum_decomposed} we have:
	
	\begin{align*}
	R_{\sigma,I}^T & =\max_{ a\in A(I)}\mleft\{\sum_{t=1}^T\one[\pi_i^t\in\Pi_i(\sigma)]\mleft(\hat u_{I}^t(a) -\hat u_{I}^t(\pi_i^t(I)) \mright)+\sum_{I'\in\mathcal{C}(I, a)} R_{ \sigma,I'}^T\mright\}\\
	& = \max_{ a\in A(I)}\mleft\{\sum_{t=1}^T\one[\pi_i^t\in\Pi_i(\sigma)]\hat u_{I}^t(a)+\sum_{I'\in\mathcal{C}(I, a)} R_{ \sigma,I'}^T\mright\}-\sum_{t=1}^T\one[\pi_i^t\in\Pi_i(\sigma)]\hat u_{I}^t(\pi_i^t(I))\\
	& \leq \max_{ a\in A(I)}\mleft\{\sum_{t=1}^T\one[\pi_i^t\in\Pi_i(\sigma)]\hat u_{I}^t(a)\mright\}+\max_{ a\in A(I)}\mleft\{\sum_{I'\in\mathcal{C}(I, a)} R_{ \sigma,I'}^T\mright\}-\sum_{t=1}^T\one[\pi_i^t\in\Pi_i(\sigma)]\hat u_{I}^t(\pi_i^t(I))\\
	& = \hat R_{\sigma,I}^T +\max_{ a\in A(I)}\mleft\{\sum_{I'\in\mathcal{C}(I, a)} R_{ \sigma,I'}^T\mright\}.
	\end{align*}
	By starting from $I$ and applying the above equation inductively, we obtain the result.
\end{proof}

\subsection{Proofs for Section~\ref{sec:icfr}}

\begin{lemma}\label{lemma:unique_seq}
	For any $I\in\I_i$ and $t = 1 , \ldots , T$, if it is the case that $\pi_i^t\notin\Pi_i(I)$, then the sequence $\sigma_I^t$ defined by \textsc{SampleInternal} exists and is unique.
\end{lemma}
\begin{proof}
	It is enough to proceed from infoset $I$ towards the root of the tree.
	Eventually, the procedure reaches an infoset $I'\in \I_i$ such that $\pi_i^t\in \Pi_i(I')$.
	Then, $\sigma_I^t$ is identified by the pair $(I',\pi_i^t(I'))$.
\end{proof}

\goingUp*
\begin{proof}\allowdisplaybreaks
	By hypothesis and since the action space $A(I)$ is finite we have that 
	\[
	\sum_{a\in A(I)} \hat R_{(I,a),J}^T=o(T).
	\]
	Moreover,
	\begin{align*}
	\sum_{a\in A(I)}\hat R_{(I,a),J}^T & = \sum_{a\in A(I)} \max_{\hat a\in A(J)}\mleft\{ \sum_{t=1}^T \one[\pi_i^t\in\Pi_i(I,a)] \mleft( \hat u_J^t(\hat a) - \hat u_J^t(\pi_i^t(J))	 \mright)  \mright\}\\
	& \geq \max_{\hat a\in A(J)}\mleft\{ \sum_{t=1}^T\sum_{a\in A(I)}\one[\pi_i^t\in\Pi_i(I,a)] \mleft( \hat u_J^t(\hat a) - \hat u_J^t(\pi_i^t(J))	 \mright) \mright\}\\
	& = \max_{\hat a\in A(J)}\mleft\{ \sum_{t=1}^T\one[\pi_i^t\in\Pi_i(I)] \mleft( \hat u_J^t(\hat a) - \hat u_J^t(\pi_i^t(J))	 \mright) \mright\}\\
	& = \max_{\hat a\in A(J)}\mleft\{ \sum_{t=1}^T\one[\pi_i^t\in\Pi_i(\sigma(I))] \mleft( \hat u_J^t(\hat a) - \hat u_J^t(\pi_i^t(J))	 \mright) \mright\}\\
	& = \hat R_{\sigma(I),J}^T.
	\end{align*}
	This concludes the proof.
\end{proof}

\icfrTh*
\begin{proof}
	By Theorem~\ref{th:efce}, in order to converge to an EFCE, it is enough to minimize the trigger regrets $R_{\sigma}^T$ for each player $i\in\pl$ and sequence $\sigma=(I,a)\in\Sigma_i$.
	This can be done by minimizing the subtree regrets $R_{\sigma,I}^T$ via the minimization of laminar subtree regrets $\hat R^T_{\sigma,I}$ for each sequence $\sigma=(J,a)\in\Sigma_i$ and infoset $I\in\mathcal{C}^\star(J)$ (Lemma~\ref{lemma:regret_bound}).

	For any infoset $I\in\I_i$, the laminar subtree regrets $\hat R_{\sigma,I}^T$ are partitioned in three groups on the basis of the trigger sequence $\sigma$:
	
	\begin{itemize}
		\item {\em Group 1}: $\sigma=(I,a)\in\Sigma_i$.
		Laminar subtree regrets belonging to this group are updated at rounds $t$ such that $\pi_i^t\in\Pi_i(I)$, otherwise they remain unchanged.
		Therefore, they are only updated when the strategy at $I$ is recommended
		by the internal-regret minimizer $\rmi_I$, which guarantees $\hat R_{\sigma,I}^T=o(T)$~\citep{cesa2006prediction}.
		
		\item {\em Group 2}: $\sigma=(J,a)\in\Sigma_i$ is such that $J\preceq I$, $J\neq I$, and $a$ is \emph{not} on the path from $J$ to $I$ ({\em i.e.}, for any $\pi_i\in\Pi_i$, $\pi_i(J)=a$ implies $\pi_i\notin\Pi_i(I)$).
		The sequence $\sigma_I^t$ is defined as a sequence compatible with $\pi_i^t$ and belonging to $\Sigma_i^c(I)$.
		By Lemma~\ref{lemma:unique_seq}, for each $I\in\I_i$ and $t = 1,\ldots, T$, $\sigma_I^t$ exists and is unique.
		Then, at most one laminar subtree regret term of Group 2 is updated at each round $t$, otherwise they are left unchanged.
		Whenever one of these regrets is affected by the choice at $t$, the action at $I$ is selected according to the external-regret minimizer $\rme_{\sigma_I^t,I}$.
		This ensures that each laminar subtree regret belonging to this group is $o(T)$ by the known properties of no-external-regret algorithms~\citep{cesa2006prediction}.
		
		\item {\em Group 3}: $\sigma=(J,a) \in \Sigma_i$ is such that $J\preceq I$, $J\neq I$, and $a$ is on the path from $J$ to $I$ (notice that for each $J\preceq I$, $J\neq I$ one such $a$ is unique because player $i$ has perfect recall).
		Let $I'\in\I_i$ be such that $J\preceq I' \preceq I$ and $I' \in\mathcal{C}(J,a)$.
		Notice that, given $I$, $J$, and $\sigma$, one such $I'$ is unique because of the perfect recall assumption.
		By Lemma~\ref{lemma:going_up}, we know that if $\hat R_{\sigma',I}^T=o(T)$ for all $\sigma'=(I',a')\in\Sigma_i$, then it must be the case that $\hat R_{\sigma, I}^T=o(T)$ (notice that $\sigma = (J,a)$ is the same as $\sigma(I')$).
		By applying the lemma recursively, until all $\sigma'$ belong to either Group 1 or 2, we can guarantee that $\hat R_{\sigma, I}^T=o(T)$.
	\end{itemize}
	This concludes the proof.
\end{proof} 

\section{Experimental Evaluation}

Appendix~\ref{subsec:exp_games} provides a detailed description of the benchmark games used in our experiments.
Finally, Appendix~\ref{subsec:more_exp} shows additional experimental results for ICFR.

\subsection{Benchmark games}\label{subsec:exp_games}

The size (in terms on number of infosets and sequences) of the parametric instances we use as benchmark is described in~\cref{fig:dimensions}.
In the following, we provide a detailed explanation of the rules of the games.

\begin{figure}[H]
	\centering
		{\begin{minipage}{.4\textwidth}
			\small
			\setlength{\tabcolsep}{2pt}
			\begin{tabular}{c@{\hskip 6pt}ccccc}
				& $|\pl|$ & \textbf{Ranks} & \textbf{Player}& \textbf{Infosets} & \textbf{Sequences} \\[-.1mm]
				\midrule\\[-3mm]
				\multirow{6}{*}{Kuhn} & \multirow{3}{*}{3} & \multirow{3}{*}{3} & Player 1 & 12 & 25 \\[-.1mm]
				&&& Player 2 & 12 & 25 \\[-.1mm]
				&&& Player 3 & 12 & 25 \\[2mm]
				& \multirow{3}{*}{3} & \multirow{3}{*}{4} & Player 1 & 16 & 33  \\[-.1mm]
				&&&  Player 2 & 16 & 33  \\[-.1mm]
				&&&  Player 3 & 16 & 33  \\[-.1mm]
				\midrule
				\multirow{7}{*}{Goofspiel} & \multirow{2}{*}{2} & \multirow{2}{*}{3} & Player 1 & 213 & 262  \\[-.1mm]
				&&& Player 2 & 213 & 262  \\[2mm]
				& \multirow{2}{*}{2} & \multirow{2}{*}{4} & Player 1 & 8716 &  10649 \\[-.1mm]
				&&&  Player 2 & 8716 &  10649 \\[2mm]
				& \multirow{3}{*}{3} & \multirow{3}{*}{3} & Player 1 & 837 &  934 \\[-.1mm]
				&&&  Player 2 & 837 & 934  \\[-.2mm]
				&&&  Player 3 & 837 & 934  \\[-.2mm]
				\midrule
				\multirow{3}{*}{Leduc} & \multirow{3}{*}{3} & \multirow{3}{*}{3} & Player 1 & 3294 &  7687 \\[-.1mm]
				&&& Player 2 & 3294 & 7687  \\[-.1mm]
				&&& Player 2 & 3294 & 7687  \\[-.1mm]
				\bottomrule
			\end{tabular}\\[2mm]
		
			\hspace{-3mm}
			\begin{tabular}{c@{\hskip 6pt}ccccc}
			& \textbf{Grid} & \textbf{Rounds} & \textbf{Player} & \textbf{Infosets} & \textbf{Sequences} \\[-.1mm]
			\midrule\\[-3mm]
			\multirow{2}{*}{Battleship} & \multirow{2}{*}{$(2,2)$} & \multirow{2}{*}{$3$} & Player 1 & 1413 & 2965 \\[-.1mm]
			&  &  & Player 2 &  1873 & 4101 \\[-.1mm]
			\bottomrule
		\end{tabular}
		\end{minipage}}
		\caption{The size of our parametric game instances in terms of number of sequences and infosets for each player of the game.}
		\label{fig:dimensions}
\end{figure}

\paragraph{Kuhn poker}
The two-player version of the game was originally proposed by~\cite{kuhn1950simplified}, while the three-player variation is due to~\cite{farina2018exante}.
In a three-player Kuhn poker game with rank $r$, there are $r$ possible cards. Each player initially pays one chip to the pot, and she/he is dealt a single private card. 
The first player may {\em check} or {\em bet} ({\em i.e.,} put an additional chip in the pot). Then, the second player can check or bet after a first player's check, or {\em fold/call} the first player's bet. If no bet was previously made, the third player can either check or bet. Otherwise, she/he has to fold or call. After a bet of the second player (resp., third player), the first player (resp., the first and the second players) still has to decide whether to fold or to call the bet. At the showdown, the player with the highest card who has not folded wins all the chips in the pot.

\paragraph{Goofspiel}
This game was originally introduced by~\cite{ross1971goofspiel}. Goofspiel is essentially a bidding game where each player has a hand of cards numbered from $1$ to $r$ ({\em i.e.}, the rank of the game). A third stack of $r$ cards is shuffled and singled out as prizes. 
Each turn, a prize card is revealed, and each player privately chooses one of her/his cards to bid, with the highest card winning the current prize. In case of a tie, the prize card is discarded. After $r$ turns, all the prizes have been dealt out and the payoff of each player is computed as follows: each prize card’s value is equal to its face value and the players’ scores are computed as the sum of the values of the prize cards they have won.
We remark that due to the tie-breaking rule that we employ, even two-player instances of the game are general-sum. All the Goofspiel instances have {\em limited information}, {\em i.e.}, actions of the other players are observed only at the end of the game. This makes the game strategically more challenging, as players have less information regarding previous opponents' actions.

\paragraph{Leduc}
We use a three-player version of the classical Leduc hold'em poker introduced by~\citet{southey2005bayes}.
In a Leduc game instance with $r$ ranks the deck consists of three suits with $r$ cards each. As the game starts players pay one chip to the pot. There are two betting rounds. In the first one a single private card is dealt to each player while in the second round a single board card is revealed. The maximum number of raise per round is set to two, with raise amounts of 2 and 4 in the first and second round, respectively.

Battleship is a parametric version of the classic board game, where two competing fleets take turns at shooting at each other. For a detailed explanation of the Battleship game see the work by~\cite{farina2019correlation} that introduced it. Our instance has loss multiplier equal to $2$, and one ship of length $2$ and value $1$ for each player

\subsection{Additional results}\label{subsec:more_exp}

We provide detailed results on the convergence of ICFR in terms of players' incentives to deviate from the obtained empirical frequency of play $\bar \mu^T$.
For EFCEs, these incentives correspond to the maximum deviation of each player, as defined in the outer maximization in Equation~\eqref{eq:efce_simp_eps}.
Intuitively, for each player, this represents the maximum utility any trigger agent for that player could gain by deviating from the point in which it gets triggered onwards.

We do not only consider deviations as prescribed by EFCE, but we also show results for other solutions concepts involving correlation in EFGs, namely EFCCEs and NFCCEs (see Appendix~\ref{appendix:efce} for their informal description, while for their formal definitions the reader can refer to~\citet{farina2019coarse}).
As for EFCCEs, the players' incentives to deviate are defined in a way similar to EFCE, using the definition of trigger agent suitable for EFCCEs (see~\citep{farina2019coarse}).
Instead, for NFCCEs, each player's incentive to deviate corresponds to the utility she/he could gain by playing the best normal-form plan given $\bar{\mu}^T$.
We recall that the following relation holds: EFCE $\subseteq$ EFCCE $\subseteq$ NFCCE.

%
% Moreover, \citet{farina2019coarse} introduce a bilinear saddle-point formulation for finding an EFCCE efficiently (in polynomial time) in two-player general-sum EFGs without chance.
%
% The same formulation can be adapted to the NFCCE problem for the easier setting of two-player games.

In Figures~\ref{fig:dev_bs}~-~\ref{fig:dev_leduc_3pl_3ranks}, we report players' incentives to deviate obtained with ICFR for EFCE (Left), EFCCE (Center), and NFCCE (Right).
As the plots show, the convergence rate is similar for the three cases, electing ICFR as an appealing algorithm also for EFCCEs and NFCCEs.
This is the first example of algorithm computing $\epsilon$-EFCCEs efficiently in general-sum EFGs with more than two players and chance.
\citet{celli2019learning} propose some algorithms to compute $\epsilon$-NFCCEs in general-sum EFGs with an arbitrary number of players (including chance).
Our algorithm outperforms those of~\citet{celli2019learning}, since the latter cannot reach a $0.1$-NFCCE in less than $24$h on a Leduc instance with 1200 total infosets and a one-bet maximum per bidding round.
Instead,~\icfr~reaches $\epsilon=0.1$ in around $9$h on an arguably more complex Leduc instance ({\em i.e.}, more than 9k total infosets and a two-bet maximum per round).

%\todo[inline]{dettagli NFCE e EFCE \\ * convergenza per singolo giocatore}

\begin{figure}[H]
	\centering
	\hspace{-6mm}
	{\begin{minipage}{5cm}\centering
		\includegraphics[width=1.1\textwidth]{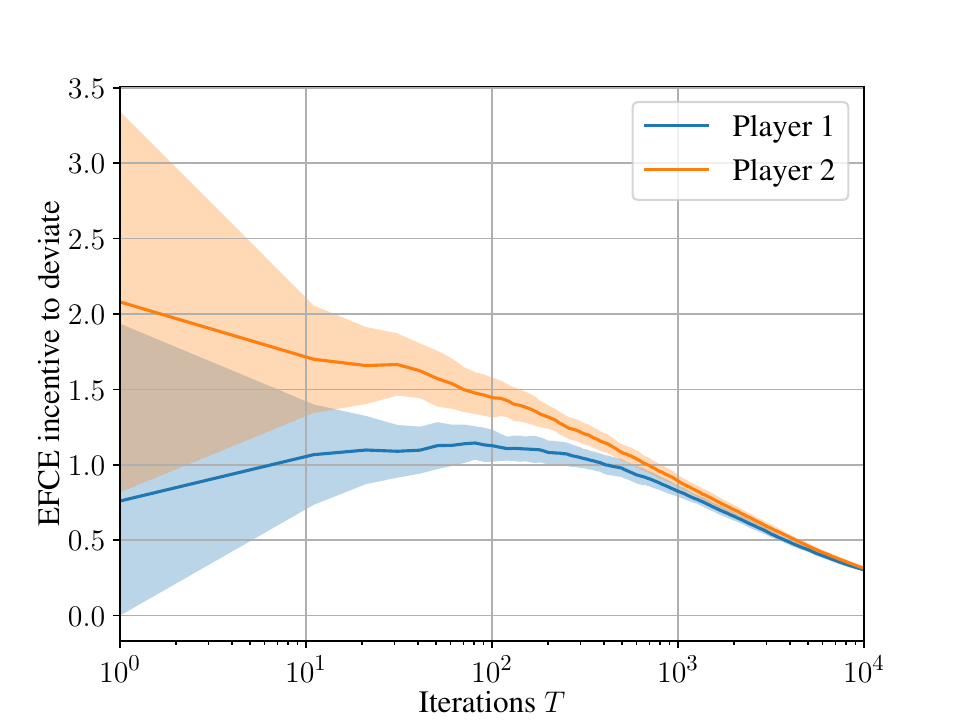}
	\end{minipage}}
	{\begin{minipage}{5cm}
		\includegraphics[width=1.1\textwidth]{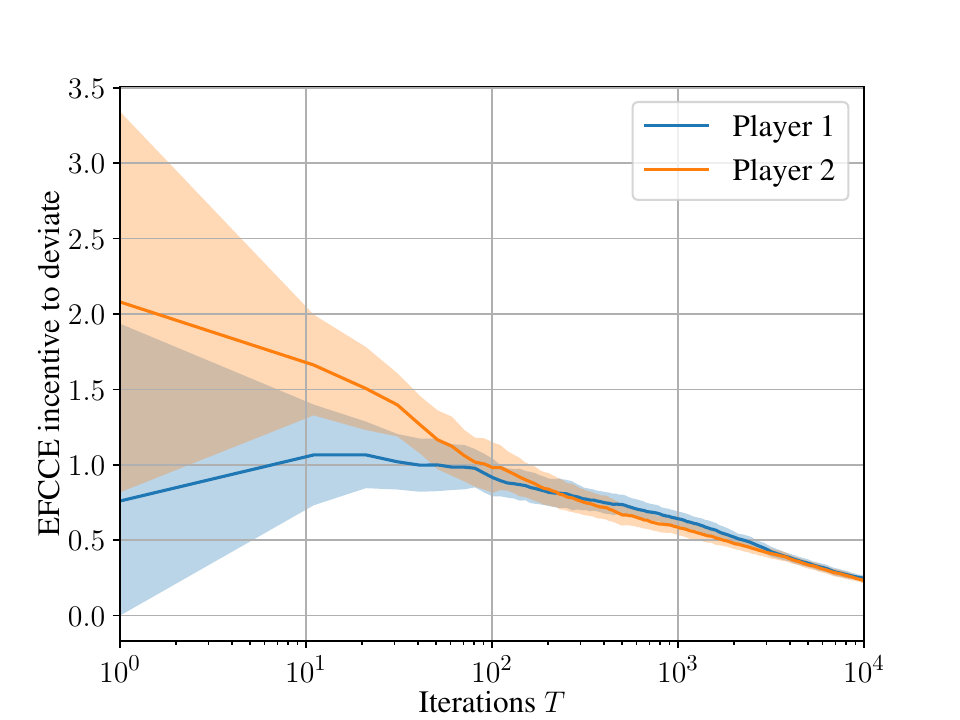}
	\end{minipage}}
	\begin{minipage}{5cm}
		\includegraphics[width=1.1\textwidth]{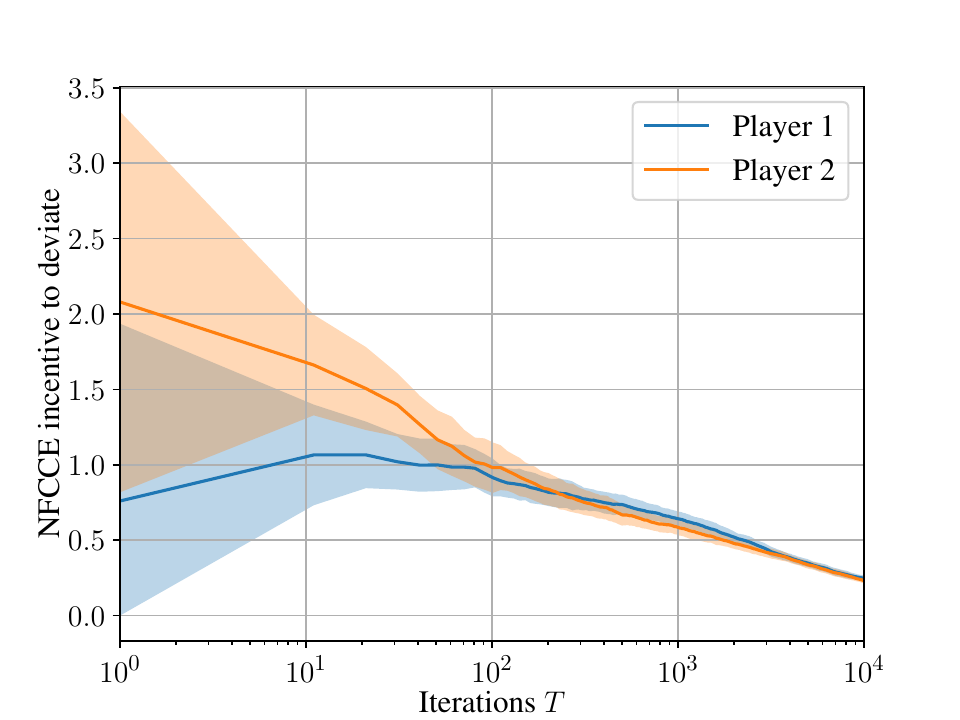}
	\end{minipage}
	\caption{Players' incentives to deviate with ICFR in the Battleship game.}
	\label{fig:dev_bs} 
\end{figure}

\begin{figure}[H]
	\centering
	\hspace{-6mm}
	{\begin{minipage}{5cm}\centering
			\includegraphics[width=1.1\textwidth]{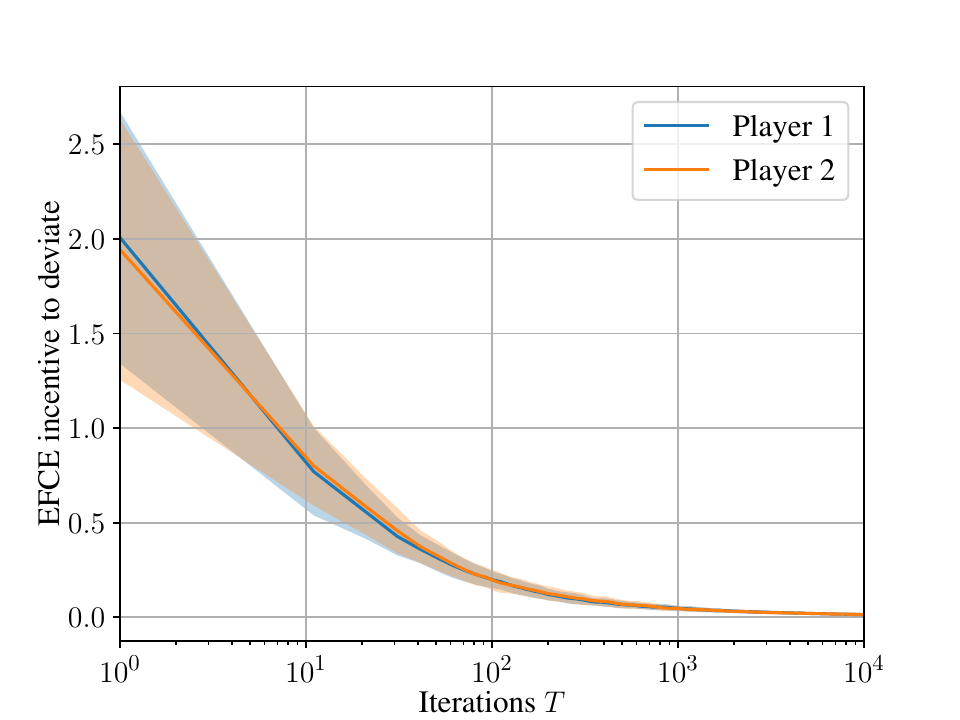}
	\end{minipage}}
	{\begin{minipage}{5cm}
			\includegraphics[width=1.1\textwidth]{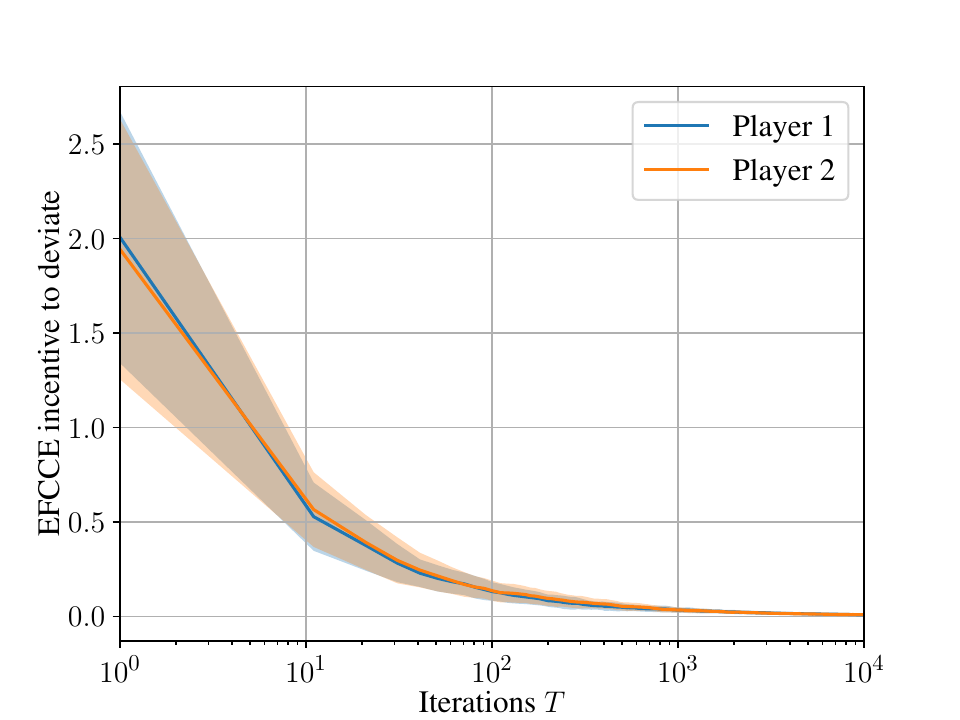}
	\end{minipage}}
	\begin{minipage}{5cm}
		\includegraphics[width=1.1\textwidth]{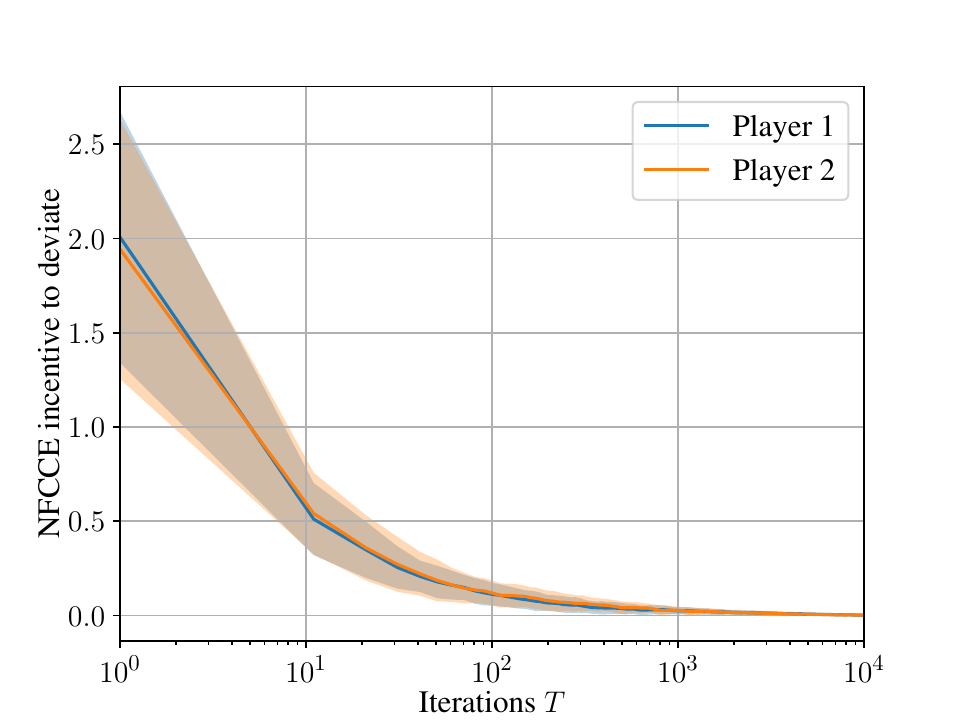}
	\end{minipage}
	\caption{Players' incentives to deviate with ICFR in two-player Goofspiel with $3$ ranks.}
	\label{fig:dev_goof_2pl_3ranks_dot_li} 
\end{figure}

\begin{figure}[H]
	\centering
	\hspace{-6mm}
	{\begin{minipage}{5cm}\centering
			\includegraphics[width=1.1\textwidth]{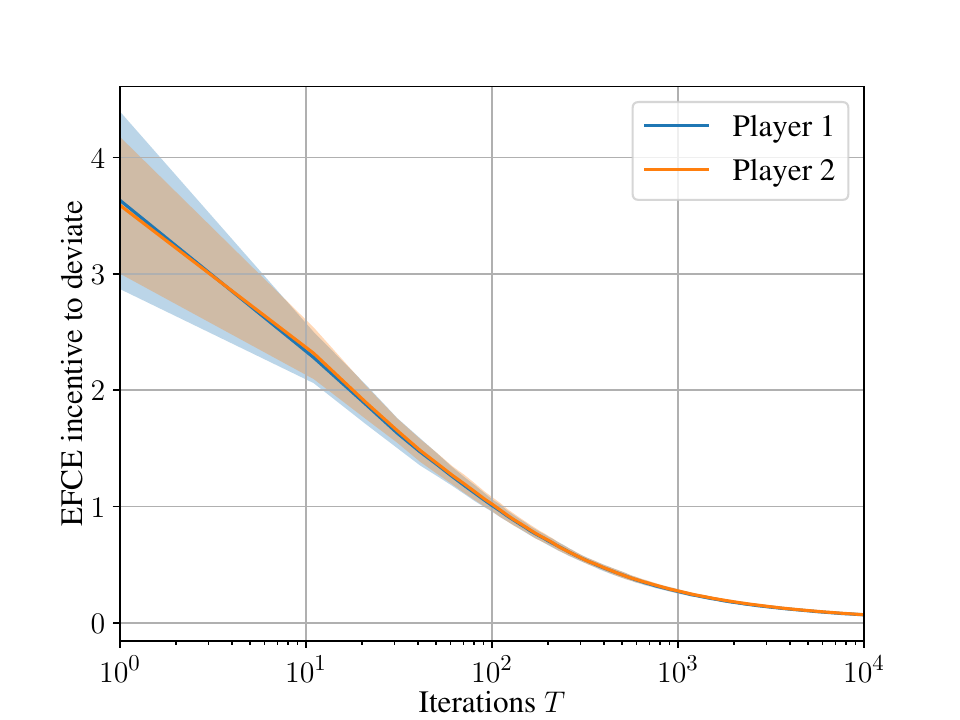}
	\end{minipage}}
	{\begin{minipage}{5cm}
			\includegraphics[width=1.1\textwidth]{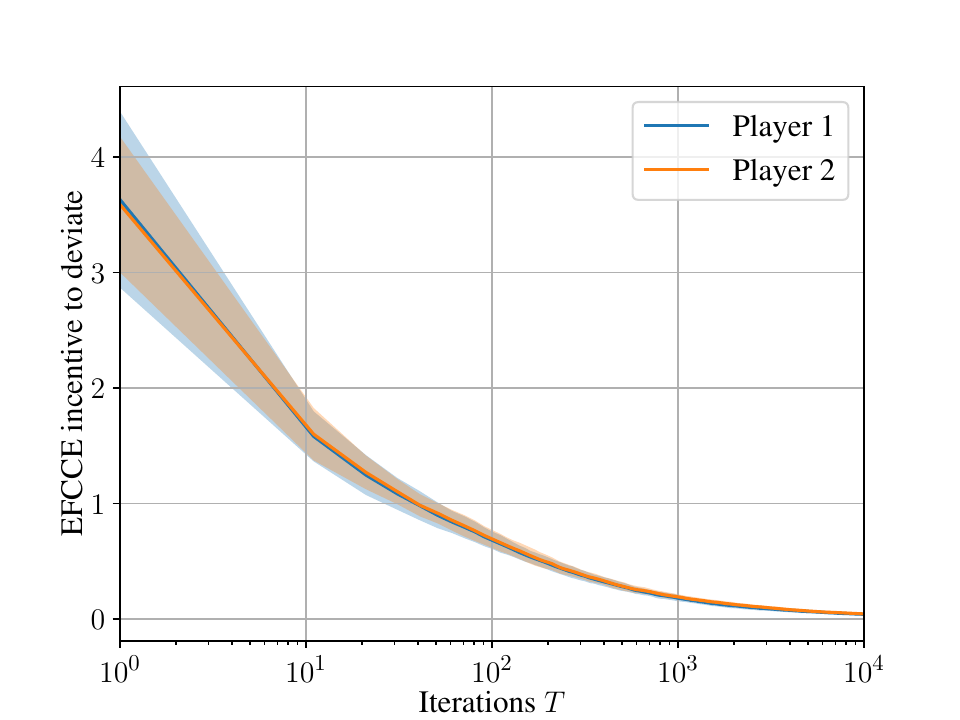}
	\end{minipage}}
	\begin{minipage}{5cm}
		\includegraphics[width=1.1\textwidth]{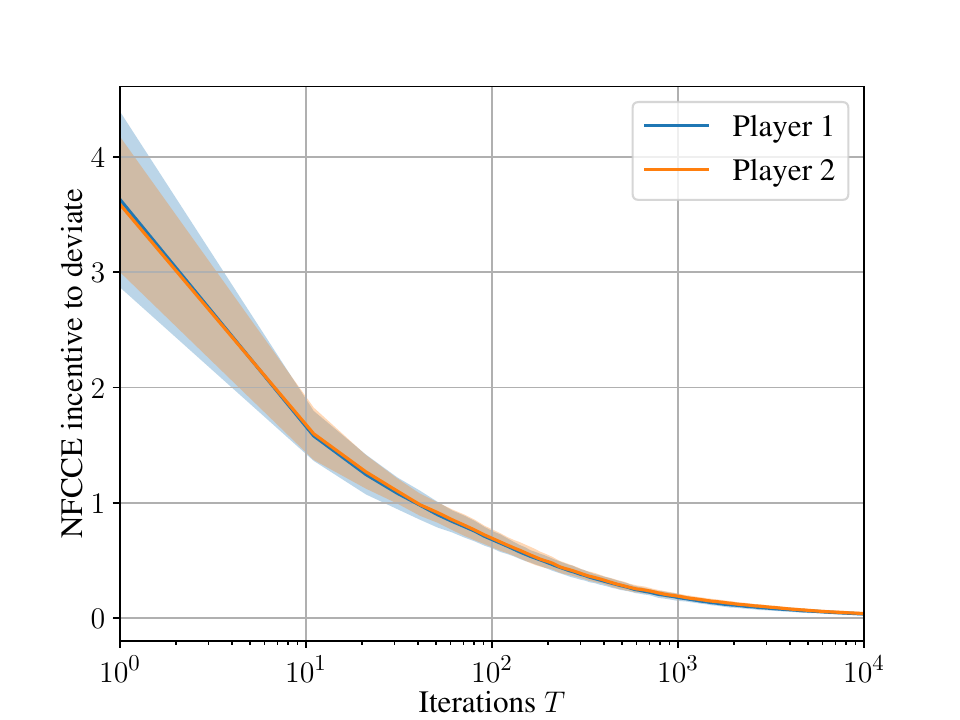}
	\end{minipage}
	\caption{Players' incentives to deviate with ICFR in two-player Goofspiel with $4$ ranks.}
	\label{fig:dev_goof_2pl_4ranks_dot_li} 
\end{figure}

\begin{figure}[H]
	\centering
	\hspace{-6mm}
	{\begin{minipage}{5cm}\centering
			\includegraphics[width=1.1\textwidth]{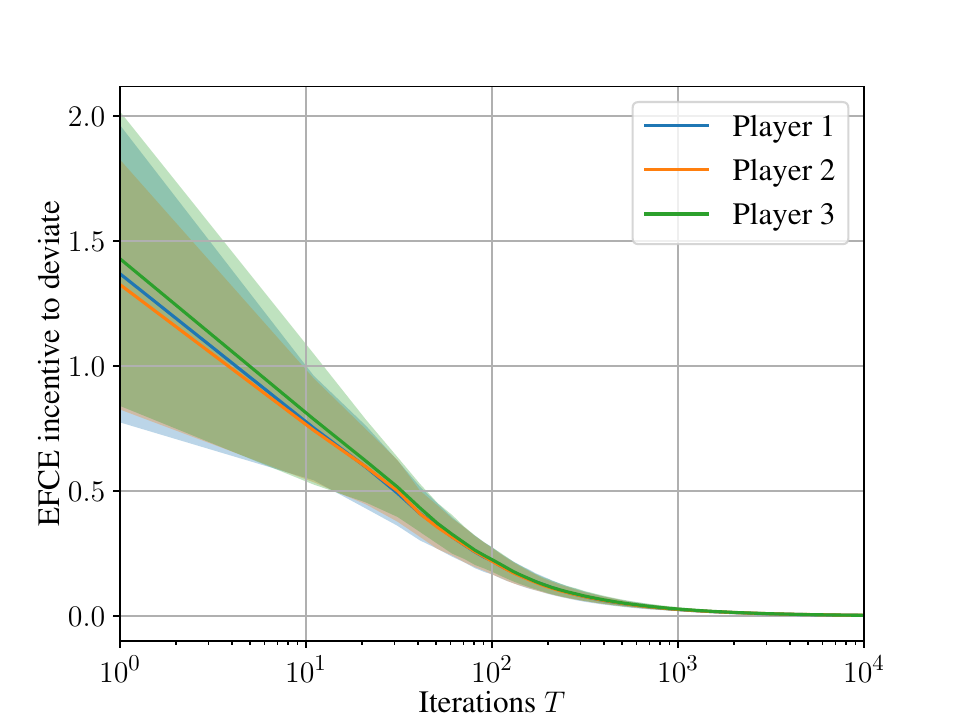}
	\end{minipage}}
	{\begin{minipage}{5cm}
			\includegraphics[width=1.1\textwidth]{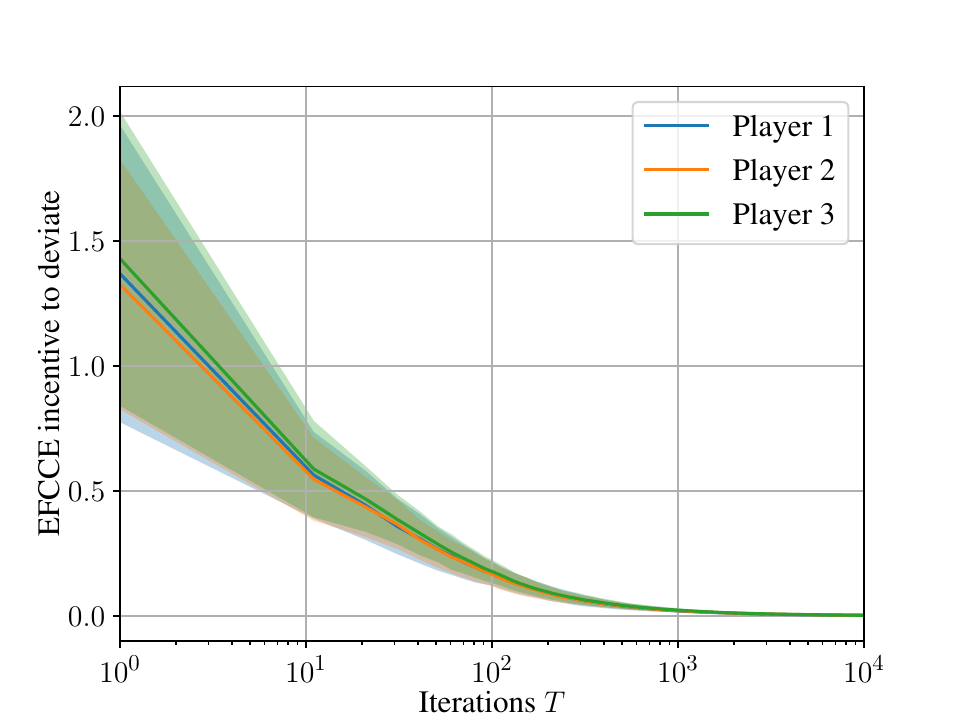}
	\end{minipage}}
	\begin{minipage}{5cm}
		\includegraphics[width=1.1\textwidth]{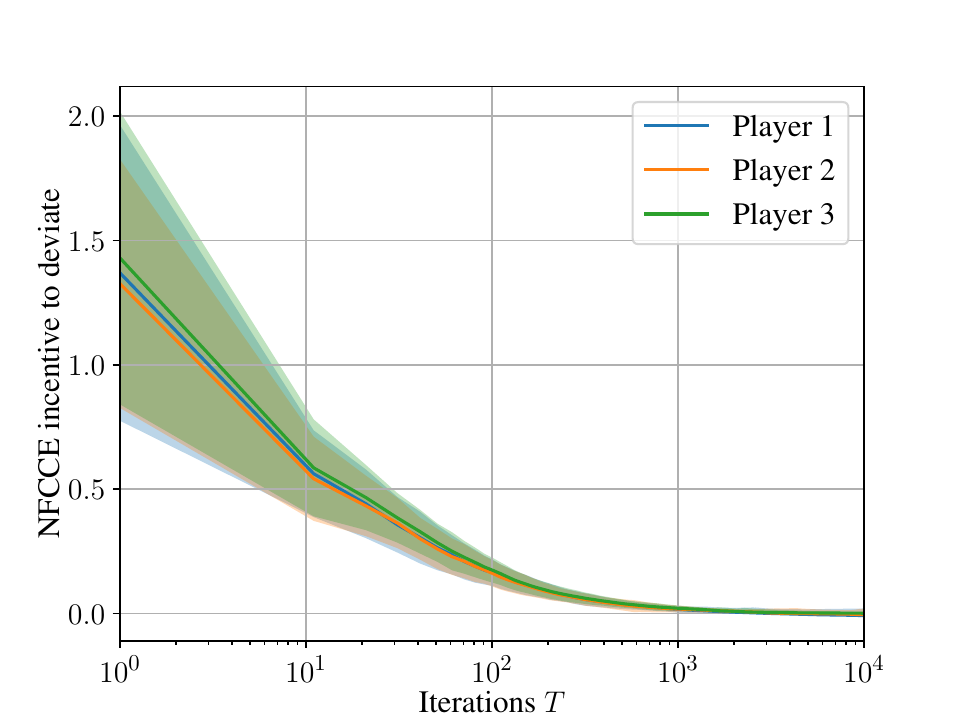}
	\end{minipage}
	\caption{Players' incentives to deviate with ICFR in three-player Goofspiel with $3$ ranks.}
	\label{fig:dev_goof_3pl_3ranks_dot_li} 
\end{figure}

\begin{figure}[H]
	\centering
	\hspace{-6mm}
	{\begin{minipage}{5cm}\centering
			\includegraphics[width=1.1\textwidth]{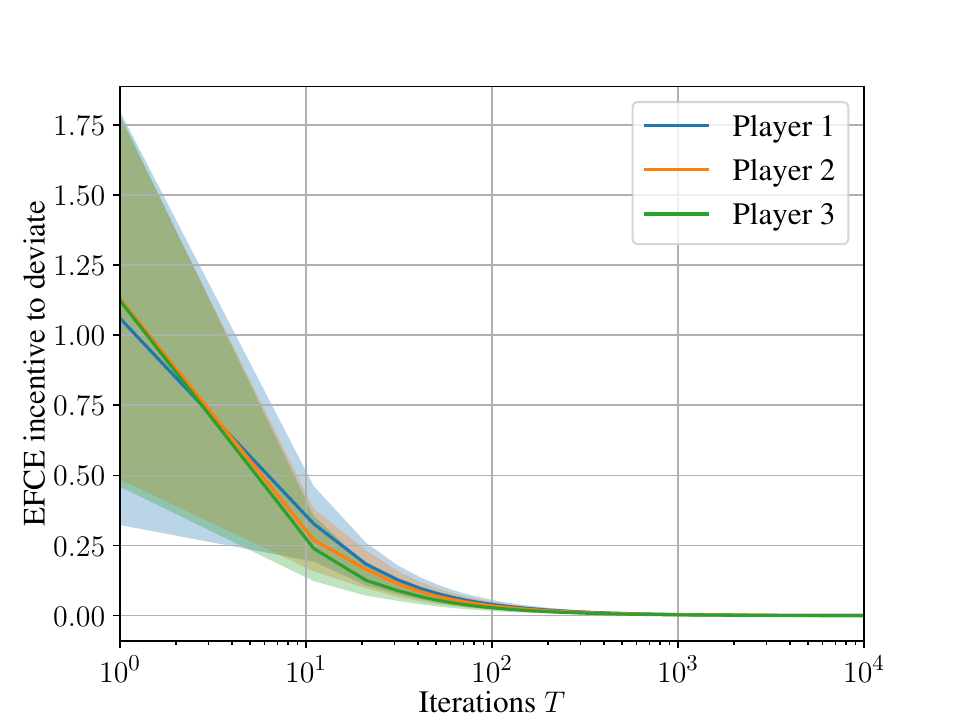}
	\end{minipage}}
	{\begin{minipage}{5cm}
			\includegraphics[width=1.1\textwidth]{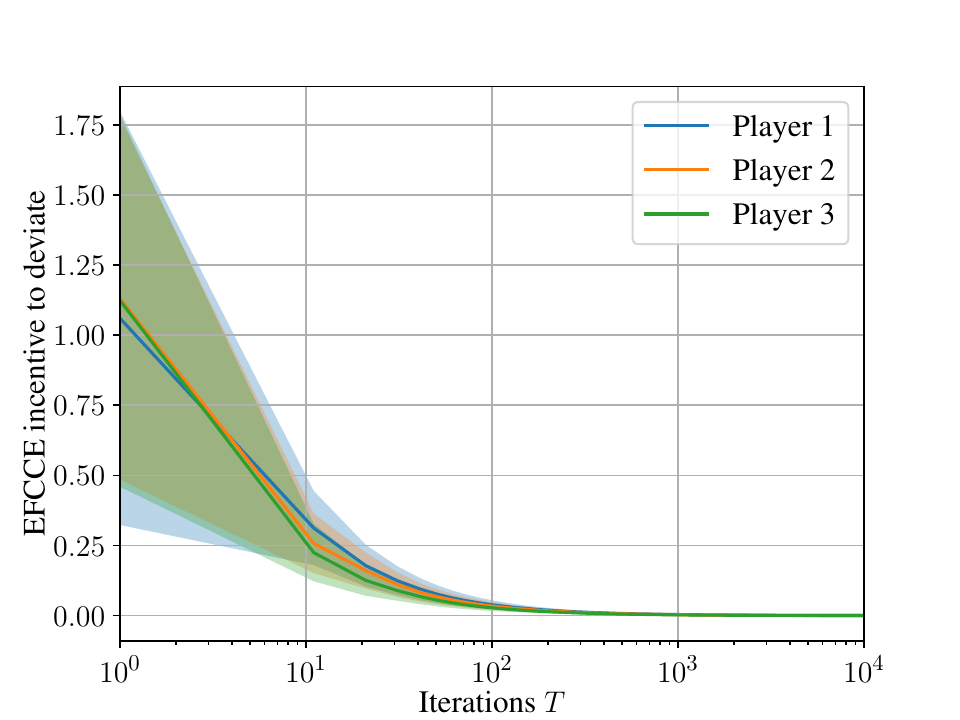}
	\end{minipage}}
	\begin{minipage}{5cm}
		\includegraphics[width=1.1\textwidth]{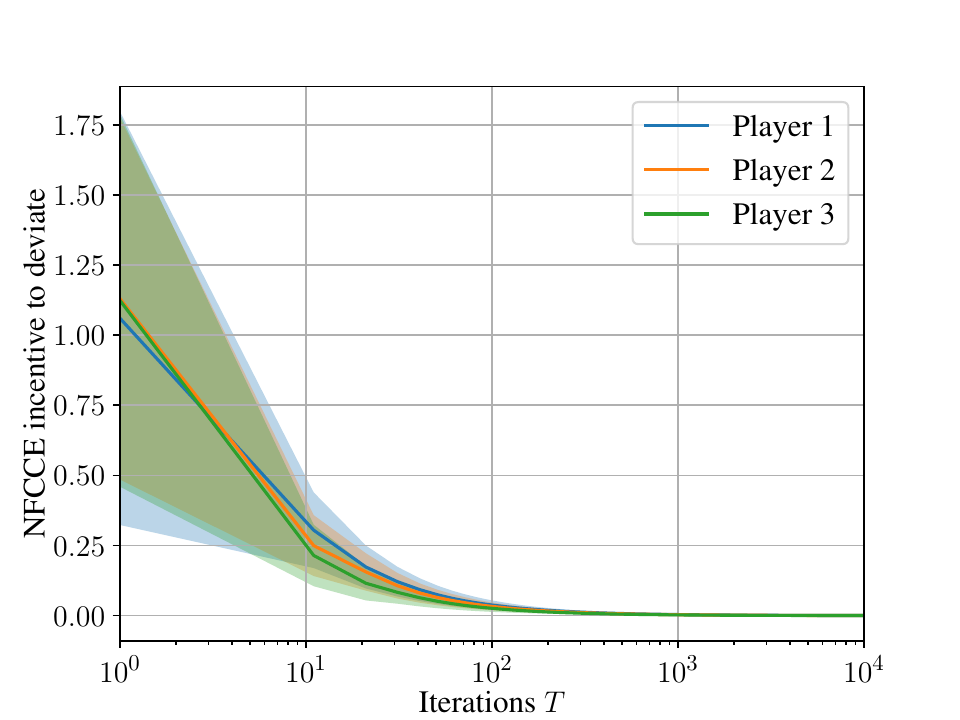}
	\end{minipage}
	\caption{Players' incentives to deviate with ICFR in three-player Kuhn Poker with $3$ ranks.}
	\label{fig:dev_kuhn_3pl_3ranks} 
\end{figure}

\begin{figure}[H]
	\centering
	\hspace{-6mm}
	{\begin{minipage}{5cm}\centering
			\includegraphics[width=1.1\textwidth]{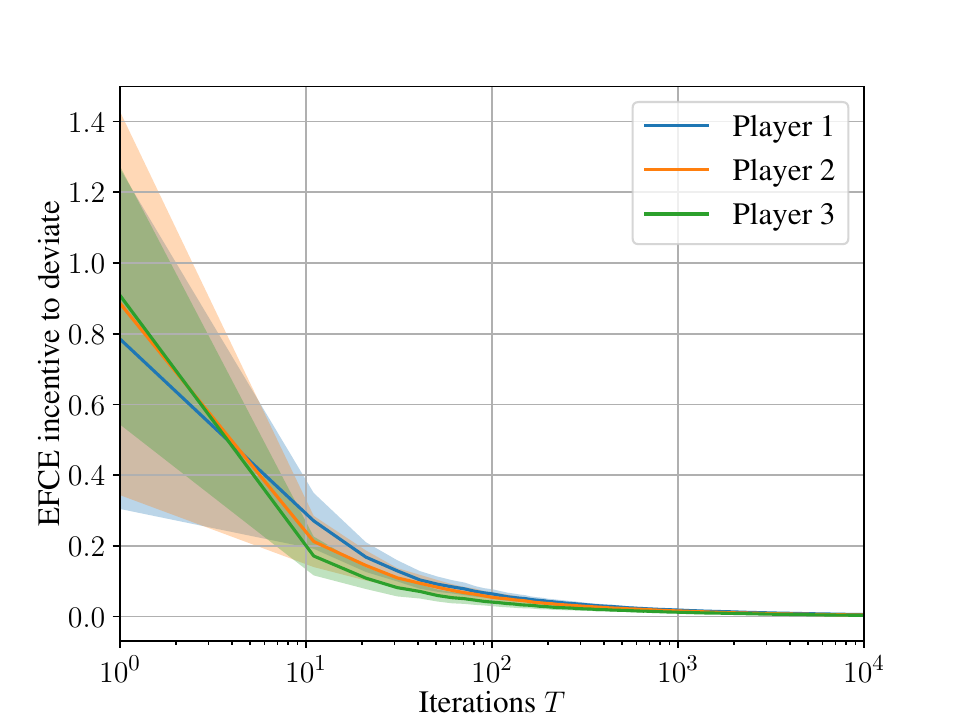}
	\end{minipage}}
	{\begin{minipage}{5cm}
			\includegraphics[width=1.1\textwidth]{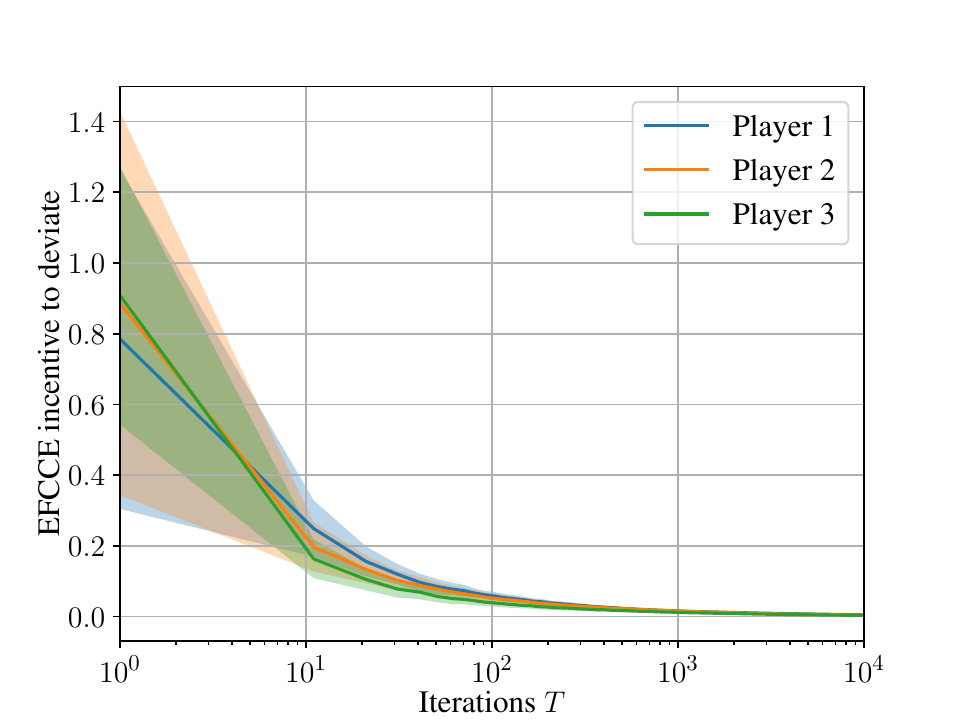}
	\end{minipage}}
	\begin{minipage}{5cm}
		\includegraphics[width=1.1\textwidth]{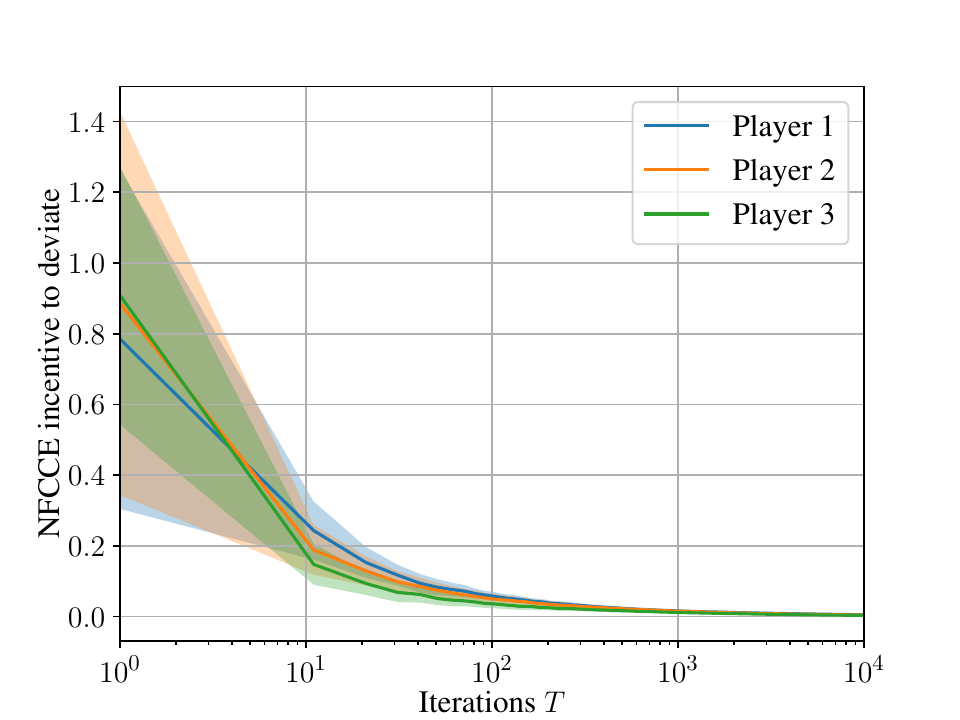}
	\end{minipage}
	\caption{Players' incentives to deviate with ICFR in three-player Kuhn Poker with $4$ ranks.}
	\label{fig:dev_kuhn_3pl_4ranks} 
\end{figure}

\begin{figure}[H]
	\centering
	\hspace{-6mm}
	{\begin{minipage}{5cm}\centering
			\includegraphics[width=1.1\textwidth]{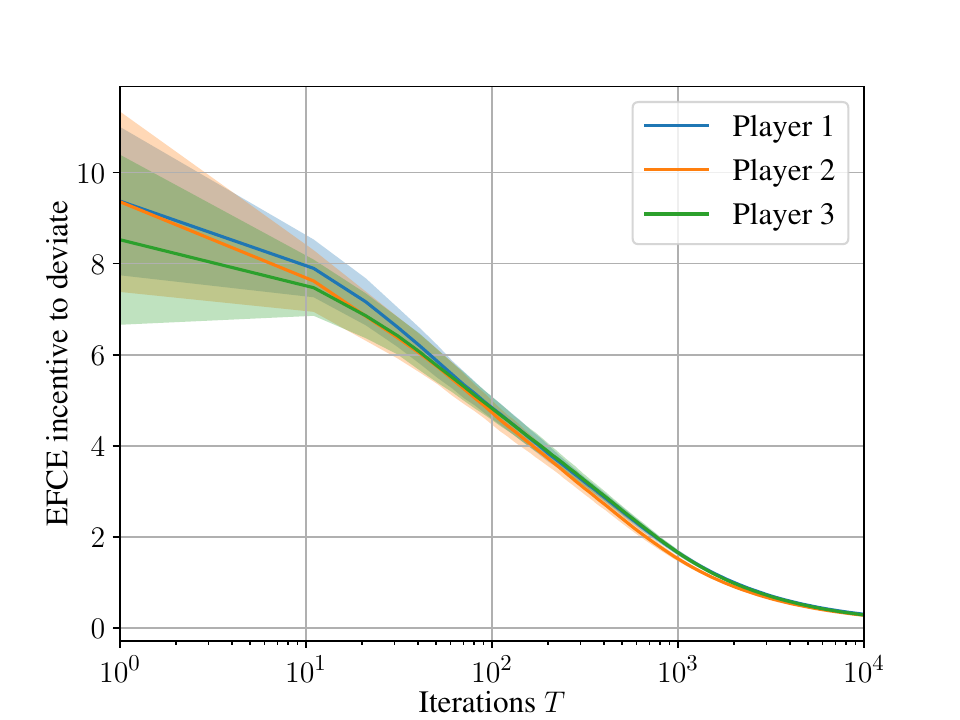}
	\end{minipage}}
	{\begin{minipage}{5cm}
			\includegraphics[width=1.1\textwidth]{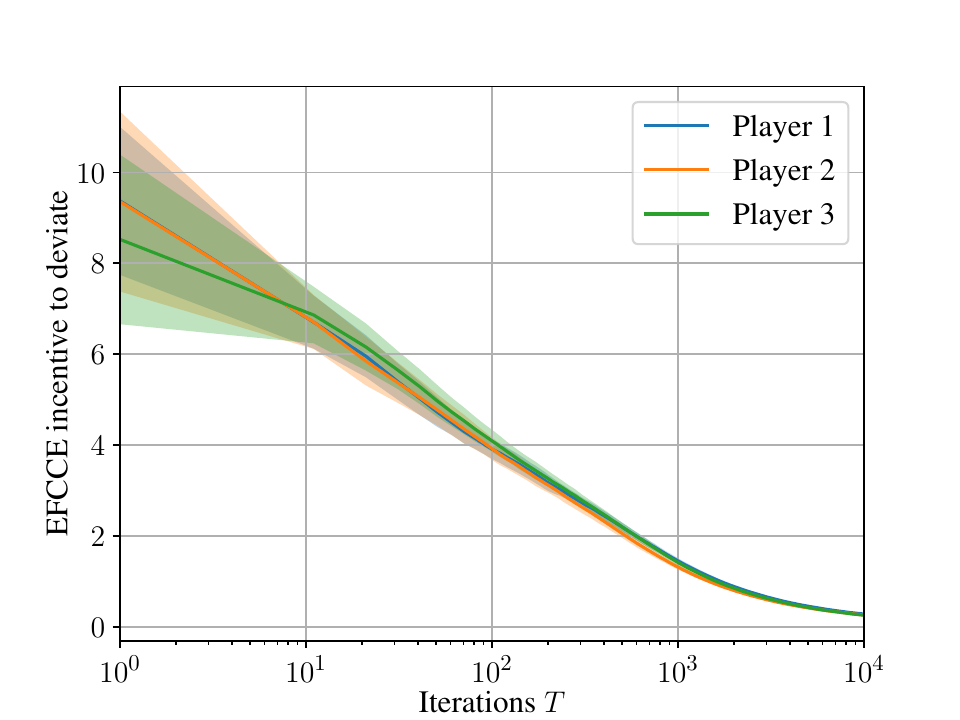}
	\end{minipage}}
	\begin{minipage}{5cm}
		\includegraphics[width=1.1\textwidth]{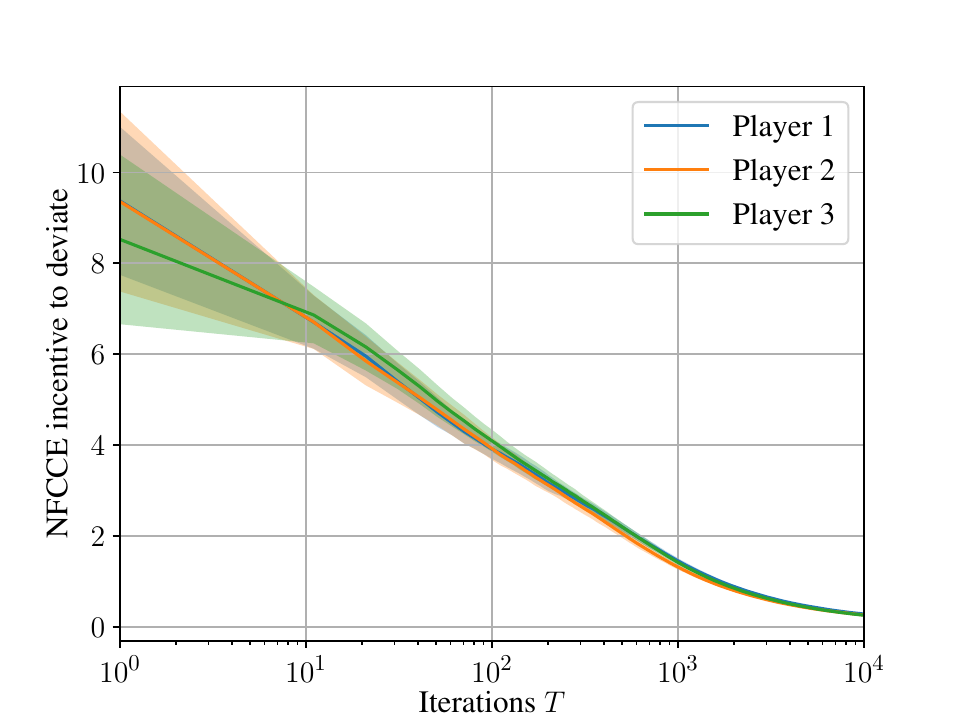}
	\end{minipage}
	\caption{Players' incentives to deviate with ICFR in three-player Leduc Poker with $3$ ranks.}
	\label{fig:dev_leduc_3pl_3ranks} 
\end{figure}

\end{document}